\documentclass{article}

\usepackage{soul}
\usepackage{url}
\usepackage{graphicx}
\expandafter\let\csname equation*\endcsname\relax
\expandafter\let\csname endequation*\endcsname\relax
\usepackage{booktabs}
\usepackage{amsthm}
\usepackage{amssymb}
\usepackage{amsmath}
\allowdisplaybreaks[4]
\usepackage{physics}
\usepackage{thmtools}
\usepackage{thm-restate}
\usepackage{threeparttable}
\usepackage{multirow}
\usepackage{makecell}
\usepackage{qcircuit}

\usepackage{framed}
\usepackage{bbm}

\usepackage{tikz}
\usetikzlibrary{positioning, shapes.geometric}

\usepackage{comment}
\usepackage{ulem}
\usepackage{hyperref}

\usepackage[accepted]{icml2024}
\usepackage[capitalize,noabbrev]{cleveref}

\renewcommand{\algorithmicrequire}{\textbf{Input:}}
\renewcommand{\algorithmicensure}{\textbf{Output:}}

\theoremstyle{plain}
\newtheorem{mydef}{Definition}

\icmltitlerunning{Quantum Algorithm for Online Exp-concave Optimization}
\begin{document}

 \twocolumn[
 \icmltitle{Quantum Algorithm for Online Exp-concave Optimization}


\begin{icmlauthorlist}
\icmlauthor{Jianhao He}{yyy}
\icmlauthor{Chengchang Liu}{yyy}
\icmlauthor{Xutong Liu}{yyy}
\icmlauthor{Lvzhou Li}{xxx}
\icmlauthor{John C.S. Lui}{yyy}
\end{icmlauthorlist}

\icmlaffiliation{yyy}{Department of Computer Science and Engineering, The Chinese University of Hong Kong, Hong Kong, China}
\icmlaffiliation{xxx}{School of Computer Science and Engineering, Sun Yat-sen University, Guangzhou, China}

\icmlcorrespondingauthor{Lvzhou Li}{lilvzh@mail.sysu.edu.cn}



\icmlkeywords{Online Exp-concave Optimization, Bandit Convex Optimization, Multi-point Bandit Feedback, Quantum Optimization Algorithms}

\vskip 0.3in
]



\printAffiliationsAndNotice{}  

\begin{abstract}
We explore whether quantum advantages can be found for the zeroth-order feedback online exp-concave optimization problem, which is also known as bandit exp-concave optimization with multi-point feedback. 
We present quantum online quasi-Newton methods to tackle the problem and show that there exists quantum advantages for such problems. 
Our method approximates the Hessian by quantum estimated inexact gradient and can achieve $O(n\log T)$ regret with $O(1)$ queries at each round, where $n$ is the dimension of the decision set and $T$ is the total decision rounds. Such regret improves the optimal classical algorithm by a factor of $T^{2/3}$.

\end{abstract}

\section{Introduction}
\label{secQONsetting}
In this paper, we consider the problem of online convex optimization for exp-concave loss functions, which is also known as online exp-concave optimization. Online convex optimization is an important framework in online learning, and particularly useful in sequential decision making problems, such as online routing \cite{awerbuch2008online}, portfolio selection \cite{Hazan2007}, and recommendation systems \cite{hazan2012projection}. 
Online convex optimization can be represented as a $T$ round iterative game between a player and an adversary.
At every round $t\in [T]$, the player generates a decision $x_t$ from a fixed and known convex set $\mathcal{K}\subseteq\mathbb{R}^n$, where $n$ is the dimension of the decision set.
The adversary observes $x_t$ and chooses a convex loss function $f_t:\mathcal{K}\to \mathbb{R}$, with which the player suffers a loss of $f_t(x_t)$.
Through this sequential process, some information about the loss function $f_t$ is revealed to the player as feedback. The goal of the player is to minimize his \textit{regret}, which is defined as
\[
    R(T) \triangleq \sum_{t=1}^T f_t(x_t)-\min_{x\in\mathcal{K}}\sum_{t=1}^T f_t(x).
\]

A good algorithm/strategy of the player should have a sublinear regret (that is, its regret is sublinear as a function of $T$) since this implies that as $T$ grows, the accumulated loss of the algorithm converges to that with the best fixed strategy in hindsight \cite{Hazan16,lattimore2020b}. 

We study the setting of exp-concave loss functions. 
This setting is also know as online exp-concave optimization (OXO). Note that OXO is useful in many machine learning problems, such as online supervised learning problems \cite{rakhlin2015online, bartlett2015minimax, gaillard2019uniform}, the portfolio selection problem \cite{luo2018efficient, mhammedi2022damped, jezequel2022efficient}.  
The loss functions are assumed to be exp-concave, which is defined as: 
\begin{mydef}[Exp-concave functions]
    A function $f:\mathcal{K}\to \mathbb{R}$ is $\alpha$-exp-concave over a convex set $\mathcal{K}$ if $g:\mathcal{K}\to \mathbb{R}, g=e^{-\alpha f(x)}$ is a concave function, where $\alpha \geq 0$.
\end{mydef}

Typical algorithms for this setting, such as the online Newton step proposed by Hazan \yrcite{Hazan2007}, can guarantee a regret bound of $O(n \log T)$ after $T$ rounds. 
But the online Newton step method requires full information feedback, which means that the loss function is revealed to the player as feedback in each round, so that the player can obtain an exact gradient to approximate the Newton's direction. 
However, in many applications, the loss function is not always easy to obtain.
When only partial information is revealed as feedback, the player has to make a prediction with limited information, which will often lead to worse regret.
Settings with such feedback are called zeroth-order feedback settings, and the feedback is modeled as a zeroth-order oracle of the loss function which is only accessible to the player after making his decision in each round. 
The zeroth-order feedback setting is well studied for the online convex optimization problem, as shown in the extended related works in Appendix \ref{Subsec:ERWOCO}.
But for online exp-concave optimization problems, it remains open. 
Liu et. al. \yrcite{liu2018online} proposed an online Newton step algorithm with estimated gradient, where the gradient is estimated by an unbiased estimator, and achieved an expected regret of $O\left(n T^{2/3} \log T\right)$.
However, there are still much room for improvement as the existing OXO work of zeroth-order feedback setting only achieves sub-optimal regret.

In the past decade, quantum computing techniques have been applied to accelerate many optimization problems, readers are referred to the extended related works in Appendix \ref{Subsec:ERWQO} for a briefly review. 
Online convex optimization problems can also be addressed using quantum computing.
In 2022, He et. al. \yrcite{he2022quantum} studied the quantum algorithm for online convex optimization of zeroth-order feedback setting, which achieves a regret bound of $O(\sqrt{T})$ that is independent of the dimension of the decision set. 
Under the same setting, this outperforms the optimal classical algorithms which have a square root dependence of the dimension.
Motivated by this, in this paper, we present quantum algorithms for the online exp-concave optimization problem with the multi-query bandit setting and we explore whether quantum advantages can be achieved. 

\paragraph{Problem Setting.} Here, an algorithm is allowed to query the zeroth-order oracle multiple times after committing a prediction so to receive feedback in each round.
A classical zeroth-order oracle $O_f$ to a loss function $f$, queried with a vector $x\in\mathcal{K}$, outputs $O_f(x)=f(x)$. 
A quantum zeroth-order oracle $Q_f$ is a unitary transformation that maps a quantum state $\ket{x}\ket{q}$ to the state $\ket{x}\ket{q+ f(x)}$, where $\ket{x}$, $\ket{q}$ and $\ket{q+ f(x)}$ are basis states corresponding to the floating-point representations of $x$, $q$ and $q+f(x)$ respectively. 
Moreover, given the superposition input $\sum_{x,q}\alpha_{x,q}\ket{x}\ket{q}$,  by linearity, the quantum oracle will output the state  $\sum_{x,q}\alpha_{x,q}\ket{x}\ket{q+ f(x)}$.
Note that we do not need to limit the power of the adversary, namely, the adversary in our setting can be \textit{completely adaptive}, and the adversary can choose $f_t$ after observing the player's choice $x_t$. 
We only assume that both the player and the adversary are quantum enabled, which means that the adversary returns a quantum oracle as feedback and the player can use a quantum computer and query the oracle with a superposition input.

In addition, we make the following assumptions which are common in online convex optimization.
We assume the loss functions are $\alpha$-exp-concave and
G-Lipschitz continuous, (i.e. $|f_t(x)-f_t(y)| \leqslant G \|y-x\|,\quad \forall x, y \in \mathcal{K}$). We also assume that the feasible set $\mathcal{K}$ is bounded and its diameter has an upper bound $D$, that is, $\forall{x,y \in \mathcal{K}}, \|x-y\|_2 \leq D$. $\mathcal{K},D,G$ are known to the player.
In addition, to avoid the mathematical proof in this paper being too lengthy and tedious, the loss functions are assumed to be $\beta$-smooth directly. 
The case where the smoothness assumption is not satisfied can be solved by using the mollification technique as shown in \cite{he2022quantum,Childs18}.
We will show how to extend our work to non-smooth case in Section \ref{lab:SecNonSmooth}.

\paragraph{Contribution.} The contribution of this paper is threefold.

\begin{itemize}
    \item 
    We propose a quantum randomized algorithm that can achieve the regret bound $O\left(\left(DG+\frac{1}{\alpha}\right)n\log(T+\log T)\right)$, by querying the oracle $O(1)$ times in each round (Theorem \ref{TNQ}), which means that our quantum algorithm outperforms the known optimal classical algorithm. Specifically, the regret of our quantum algorithm improves the optimal classical algorithm~\cite{liu2018online} by a factor of $T^{2/3}$, as shown in Table \ref{tab:RB}. 
    \item We extend the analysis of online Newton step method for the situation that the gradient is inexact but the error is controllable and can be by $\ell_1$ norm. This is the theoretical basis for selecting the parameters of quantum circuits and the learning rate appropriately. This technique also implies that there exists a classical algorithm with finite difference method can achieve a logarithm regret by querying the oracle $O(n)$ times in each round, as finite difference method can give an error controllable estimated gradient bounded by $\ell_1$ norm as well.
    \item We generalize our methods and propose quantum adaptive gradient methods and quantum online Newton methods with Hessian update respectively. The analysis of quantum adaptive gradient methods is closely related to online Newton methods. For the later extension, we show how to extend the quantum gradient estimation algorithm to the quantum Hessian estimation algorithm for the case that the quantum first-order oracle is accessible. Noticing that there is no `real' online Newton method yet, namely, there is no classical online algorithm which use the Hessian to design the update rule. Methods which approximate the Newton direction by the gradient information are generally called the quasi-Newton methods in offline optimization problems, but the word `quasi' is usually omitted in online optimization problems. To fill this void, we propose a framework of online Newton method with Hessian update, and show that it can achieve logarithm regret with appropriate parameters (Theorem \ref{TNHQ}). We also generalize our methods to the non-smooth case (Theorem \ref{TNQs}).
\end{itemize}

\begin{table*}[htbp]
    \centering
    \begin{threeparttable}
    \caption{Regret bound for online exp-concave optimization.}
    \begin{tabular}{l c c}
        \toprule
         Paper & Feedback & Regret \\
        \midrule
        \cite{Hazan2007} & Full information (first-order) & $O\left(DGn \log T\right)$ \\
        \cite{liu2018online} & Zeroth-order & $O\left(DGn T^{2/3} \log T\right)$ \\
        This work & Zeroth-order & $O\left(DGn \log (T+\log T)\right)$\\
        \bottomrule
    \end{tabular}
    \begin{tablenotes}
        \item[-] $D$ is the diameter upper bound of  the feasible set, $G$ is the Lipschitz constant, $n$ is the dimension of the feasible set, $T$ is the total decision rounds.
    \end{tablenotes}
    \label{tab:RB}
    \end{threeparttable}
\end{table*}

The remainder of this paper is organized as follows. Section \ref{lab:SecPre} introduces the basic concept of quantum computing and basic framework of online Newton step method of \cite{Hazan2007}. 
Section \ref{lab:SecQON} presents the quantum online quasi-Newton method, where the parameter choosing and the regret analysis of the algorithm are given in Section \ref{lab:SecRegret}. 
Section \ref{lab:SecExtension} gives the idea about extending our work to quantum adaptive gradient method. 
Section \ref{lab:SecON} gives the online Newton method with Hessian update and its regret analysis, accompanying with the quantum Hessian estimation algorithm.
Section \ref{lab:SecNonSmooth} shows how to extend our work to non-smooth loss functions.
Finally, we conclude with a discussion in Section \ref{lab:SecConclusion}.
The extended related works of online convex optimization and quantum optimization are placed in Appendix \ref{lab:SecERW}.
The proof details are given in Appendix \ref{lab:SecProof}.

\section{Preliminaries}
\label{lab:SecPre}
\subsection{Notations}
Quantum computing can be described with Dirac notations. We denote the computational basis of $\mathbb{C}^d$ as $\{\ket{i}\}_{i=0}^{d-1}$, where $\ket{i}$ is a $d$-dimensional vector with $1$ in the $i^{th}$ entry and $0$ in other entries. 
A $d$-dimensional quantum state can be described with an unit vector $\ket{v}=(v_1,v_2,\dots,v_d)^T=\sum_i v_i\ket{i} \in \mathbb{C}^d$ where $v_i$ is called the amplitude of $\ket{i}$ and $\sum_i \abs{v_i}^2=1$. 
The composite system of quantum systems can be described by tensor product of quantum states, that is, for $\ket{v}\in \mathbb{C}^{d_1},\ket{u} \in \mathbb{C}^{d_2}$, the composite state of this two states is $\ket{v} \otimes \ket{u} = (v_1u_1,v_1u_2,\dots.v_2u_1,\dots,v_{d_1}u_{d_2})\in \mathbb{C}^{d_1\times d_2}$. 
The notation $\otimes$ is often omitted when there is no ambiguity. The evolution of a closed quantum system can be described by unitary transformations. When we say `measure a quantum state', we mean measuring the state in computational basis unless otherwise stated. 
The measurement will give one of the state in the computational basis with the probability of the square of the magnitude of its amplitude. 
For example, if we measure $\ket{v}=\sum_i v_i\ket{i}$, we will get $i$ with probability $\abs{v_i}^2$, and the state will collapse in $\ket{i}$ after measurement, for all $i$.

We also introduce the basic notations for the optimization methods as follows.
The horizon, namely, the total decision rounds is denoted by $T$, and in order to distinguish it from the notation of matrix transpose, the notation of matrix transpose is denoted by Room letter $\mathrm{T}$.
The decision set is denoted by $\mathcal{K}\subseteq\mathbb{R}^n$, where $n$ is the dimension of the decision set and its diameter has an upper bound $D$.
The decision made by the player in round $t$ is denoted by $x_t$.
The loss function chose by the adversary in round $t$ is denoted by $f_t$, the corresponding gradient is denoted by $\nabla f_t$ and the Hessian is denoted by $H_t(f_t)$.
The Lipschitz parameters of the loss functions and its gradient functions are denoted by $G$ and $L$, respectively.
We use $\norm{\cdot}$ to present the Euclidean norm of vector, and given a positive definite matrix $A$, we use $\norm{x}_A \triangleq \sqrt{x^{\mathrm{T}} A x} $ to present the weighted norm of vector $x$.
The projection according to the weighted norm defined by the matrix $A$ is denoted as $\hat{P}^{(A_t)}_{\mathcal{K}}$:
\begin{align}
    \hat{P}^{(A_t)}_{\mathcal{K}}(y) \triangleq  \mathop{\arg\min_{x \in \mathcal{K}}} \|x - y\|_{A_t}.
\end{align}

\subsection{Basic Concepts and Frameworks}
In this subsection, we introduce basic concept of quantum circuit and the framework of classical online Newton methods.

Analogous to the way a classical computer is built from an electrical circuit containing wires and logic gates, a quantum computer is built from a quantum circuit containing wires and elementary quantum gates to carry around and manipulate the quantum information \cite{Nielsen2002}.
Thus, the quantum circuit is an efficient way to describe quantum algorithms.
The basic idea of a quantum circuit is as follows: horizontal lines are used to represent quantum bits (or quantum registers). 
In general, single lines for quantum bits (registers) and double lines for classical information. 
Symbols within a box are used to represent quantum gates and quantum measurements, with the position of the box representing the quantum bits to be executed. 
Time flows from left to right, and the quantum circuit is executed from left to right.
Note that any classical circuit can be replaced by an equivalent circuit containing only reversible elements, and thus can be simulated using a quantum circuit.

We present the online Newton methods~\cite{Hazan2007} in Algorithm \ref{ONS}, which is a basic framework for solving the online exp-concave optimization problem. 
It starts with a random decision, generates the decision of the next round by moving the current decision in the quasi-Newton direction of the current loss function and then projecting to the decision set, where the matrix $A_t$ in the update rule is updated by a rank-1 matrix generated by the exact gradient information.
With the exact gradient information, the online Newton step method can achieve logarithm regret.

\begin{algorithm}[H]
    \caption{Online Newton Step Method \cite{Hazan2007}}
    \label{ONS}
    \begin{algorithmic}[1]
        \REQUIRE Step size $\gamma$, parameter $\epsilon$.
        \ENSURE $x_1, x_2, x_3, \dots x_T$
      \STATE  Choose the initial point $x_1\in\mathcal{K}$ randomly, let $A_0=\epsilon I_n$.
        \FOR{$t=1$ to $T$}
            \STATE Play $x_t$, observe the loss function $f_t$ from the adversary.
            \STATE Update $x_{t+1}=\hat{P}^{(A_t)}_{\mathcal{K}}\left(x_t-\frac{1}{\gamma} A_t^{-1} \nabla f_t(x_t)\right)$, where $A_t=A_{t-1}+\nabla f_t(x_t)(\nabla f_t(x_t))^{\mathrm{T}}$.
        \ENDFOR
    \end{algorithmic}
\end{algorithm}

\section{Quantum Online Quasi-Newton Method}
\label{lab:SecQON}

In this section, for the OXO problem stated in Section \ref{secQONsetting}, given the total horizon $T$ and $\delta$,  we present Algorithm \ref{QONS} to produce a decision sequence $x_1, x_2, x_3, \dots, x_T$ for the player, such that with probability greater than $1-\delta$, it achieves a regret being logarithm of $T$. 
Specifically, $\delta$ is divided into a series of parameters $\rho_t$ which are intermediate parameters used to adjust the success probability of the quantum gradient estimation (Lemma \ref{QGB}).

Initially, the algorithm chooses $x_1$ randomly from $\mathcal{K}$, and then sequentially produces  $x_2, x_3, \dots, x_T$ by online Newton step method. Steps 4-11 are the process of quantum gradient estimation, where definitions of basic quantum gates/circuit is given in Appendix \ref{proof:QGE}, and the quantum circuit is illustrated in Figure \ref{QGEC}.
The quantum circuit of $Q_{F_t}$ in Step 6 is constructed by using $Q_{f_t}$ twice;
Since the evolution of a quantum operation is reversible, the $Q_F^{-1}$ can be operationally realized by reversing the inputs and outputs of $Q_F$. 
$\mathbbm{1}$ in Step 6 is the $n$-dimensional all 1's vector; 
the last register and the operation of addition modulo $2^c$ in Step 7 are used for implementing the common technique in quantum algorithm known as phase kickback which adds a phase shift related to the oracle; 
Step 8 is known as uncompute trick which recovers the ancillary registers to the initial states so that they can be reused in the next iterative.

For the update rule, if one needs to compute the invert of $A_t$ for each round $t$, it will incur a complexity of $O(n^3)$. But with the help of Sherman-Morrison-Woodbury Formula \yrcite{sherman1950adjustment,woodbury1950inverting}, 
\begin{align}
A_t^{-1} & =\left(A_{t-1}+\widetilde{\nabla} f_t(x_t)(\widetilde{\nabla} f_t(x_t))^T\right)^{-1} \nonumber \\
& = A_{t-1}^{-1} - A_{t-1}^{-1} \widetilde{\nabla} f_t(x_t) \nonumber \\ 
& \quad \times \left( 1 + (\widetilde{\nabla} f_t(x_t))^T A_{t-1}^{-1} \widetilde{\nabla} f_t(x_t)) \right)^{-1} \nonumber \\
& \quad \times  (\widetilde{\nabla} f_t(x_t))^T A_{t-1}^{-1},
\end{align}
the computation complexity can be reduced to $O(n^2)$. The initial $A_0$ is a identity matrix multiplied by a constant, thus $A_0^{-1}$ can be computed in $O(n)$.

\begin{algorithm}[H]
    \caption{Quantum Online Quasi-Newton Method (QONS)}
    \label{QONS}
    \begin{algorithmic}[1]
        \REQUIRE Step sizes $\{\eta_t\}$, parameters $\{r_t\}, \rho, \epsilon$.
        \ENSURE $x_1, x_2, x_3, \dots x_T$
      \STATE  Choose the initial point $x_1\in\mathcal{K}$ randomly, let $A_0=\epsilon I_n$.
        \FOR{$t=1$ to $T$}
            \STATE Play $x_t$, get the oracle of loss function $Q_{f_t}$ from the adversary.
            \STATE Prepare the initial state: $n$ $b$-qubit registers $\ket{0^{\otimes b},0^{\otimes b},\dots,0^{\otimes b}}$ where $b=\log_2 \cfrac{G\rho}{4\pi n^2 \beta r_t}$. 
            Prepare $1$ $c$-qubit register $\ket{0^{\otimes c}}$ where $c=\log_2{\cfrac{16\pi n}{\rho}}-1$.
            And prepare $\ket{y_0}=\cfrac{1}{\sqrt{2^n}}\sum_{a\in\{0,1,\dots,2^n-1\}}e^{\cfrac{2\pi i a}{2^n}}\ket{a}$.
            \STATE Apply Hadamard transform to the first $n$ registers.
            \STATE Perform the quantum query oracle $Q_{F_t}$ to the first $n+1$ registers, where  $F_t(u)=\cfrac{2^b}{2Gr_t} \left [f_t \left (x_t+\cfrac{r_t}{2^b} \left (u-\cfrac{2^b}{2}\mathbbm{1} \right ) \right )-f_t(x_t) \right ]$, and the result is stored in the $(n+1)$th register.
            \STATE Perform the addition modulo $2^c$ operation to the last two registers.
            \STATE Apply the inverse evaluating oracle $Q_{F_t}^{-1}$ to the first $n+1$ registers.
            \STATE Perform quantum inverse Fourier transformations to the first $n$ registers separately.
            \STATE Measure the first $n$ registers in computation bases respectively to get $m_1,m_2,\dots,m_n$.
            \STATE $\widetilde{\nabla} f_t(x_t)=\cfrac{2G}{2^b} \left (m_1-\cfrac{2^b}{2},m_2-\cfrac{2^b}{2},\dots, m_n-\cfrac{2^b}{2} \right )^{\mathrm{T}}$.
            \STATE Update $x_{t+1}=\hat{P}^{(A_t)}_{\mathcal{K}}\left(x_t-\frac{1}{\eta_t} A_t^{-1} \widetilde{\nabla} f_t(x_t)\right)$, where $A_t=A_{t-1}+\widetilde{\nabla} f_t(x_t)(\widetilde{\nabla} f_t(x_t))^{\mathrm{T}}$ .
            \STATE Bitwise erase the first $n$ registers with control-not gates controlled by the corresponding classical information of the measurement results $m_1,m_2,\dots,m_n$.
        \ENDFOR
    \end{algorithmic}
\end{algorithm}

We now analyze the query complexity of Algorithm \ref{QONS}. In each round, it needs to call the oracle twice to construct $Q_F$, and twice to perform the uncompute step $Q_F^{-1}$, so totally $4$ times for computing the gradient. Thus, $O(1)$ times for each round. Then, the error bound of the estimated gradient is shown in Lemma \ref{QGB}, and the proof is given in Appendix \ref{proof:QGE}.

\begin{restatable}{mylem}{QGB}
\label{QGB}
    In Algorithm \ref{QONS}, for all timestep $t$, if $f_t$ is $\beta$-smooth function,  then for any $r_t>0$ and $1\geq \rho_t >0$, the estimated gradient $\widetilde{\nabla} f_t(x_t)$ satisfies
    \begin{align}
         \Pr[\|{\nabla f_t(x_t)-\widetilde{\nabla} f_t(x_t)}\|_1 >8 \pi n^3 (n/\rho_t +1) \beta r_t/\rho_t] <  \rho_t. 
    \end{align}
\end{restatable}

The framework of quantum gradient estimator originates from Jordan's quantum gradient estimation method \cite{Jordan05}, but Jordan's algorithm did not give any error bound because the analysis of it was given by omitting the high-order terms of Taylor expansion of the function directly. 
In 2019, the quantum gradient estimation method with error analysis was given in \cite{Gilyen2019}, and was applied to the general convex optimization problem \cite{DeWolf18,Childs18}.
In those case, however, $O(\log{n})$ repetitions were needed to estimate the gradient within a acceptable error, which means that the query complexity to the zeroth-order oracle is $O(\log{n})$. 
We improve the analysis of the quantum gradient estimation method in our previous work \cite{he2022quantum}, and show that $O(1)$ queries is sufficient  in our problem, instead of $O(\log{n})$ repetitions. 
The key technique is, for each coordinate, the failure probability of a single repetition can be made small at the expense of a weaker quality of approximation, and then the worse approximation guarantee can be fixed by choosing a finer grid.

\section{Parameters Selection and Regret Analysis}
\label{lab:SecRegret}
In this section, we show how to choose appropriate parameters such that Algorithm \ref{QONS} guarantees $O\left(\left(DG+\frac{1}{\alpha}\right)n\log(T+\log T)\right)$ regret, which gives Theorem \ref{TNQ}. 
Firstly, in each round, we need to bound the difference between the loss suffered by the player and the function value of best fixed strategy $x^*$, but the latter one is unknown to us when analyzing the regret bound, so we set to prove a stronger property that holds for any points in the decision set as shown in Lemma \ref{QNSGB}. 
Then, let this `any point' be $x^*$ will give what we need. 
Secondly, summing up these inequalities for all $T$ round will give the regret bound, and we need to choose appropriate parameters to bound every term in $O(\log T)$ so to make the regret as small as possible.
Since the gradient estimator is not accurate, it leads to more challenging analysis.

\begin{restatable}{mylem}{QNSGB}
    \label{QNSGB}
    In Algorithm \ref{QONS}, for all timestep $t$, let $f_t:\mathcal{K} \to \mathbb{R}$ be $\alpha$-exp-concave functions with Lipschitz parameter $G$, where $\mathcal{K}$ is a convex set with diameter $D$, then for any $y\in \mathcal{K}$ and any $\eta_t\leq \min \left\{\frac{1}{8GD},\frac{\alpha}{2}\right\}$, the estimated gradient $\widetilde{\nabla} f_t(x_t)$ satisfies
    \begin{align}
        \label{QNSGB1}
        f_t(y)\geq & f_t(x_t)+\widetilde{\nabla} f_t(x_t)^{\mathrm{T}}(y-x_t) \nonumber \\ 
        &+\frac{\eta_t}{2}(y-x_t)^{\mathrm{T}} \widetilde{\nabla} f_t(x_t) \widetilde{\nabla} f_t(x_t)^{\mathrm{T}}(y-x_t) \nonumber \\
        &- \frac{9D}{8} \norm{\nabla f_t(x_t) - \widetilde{\nabla} f_t(x_t)}_1 \nonumber \\
        &- \frac{D}{16G} \norm{\nabla f_t(x_t) - \widetilde{\nabla} f_t(x_t)}^2_1.
    \end{align}
\end{restatable}

The proof is given in Appendix \ref{proof:QNSGB}. 
In the following, we show that with appropriate parameters, Algorithm \ref{QONS} can guarantee logarithm regret.

\begin{restatable}{mythe}{TNQ}
    \label{TNQ}
    Algorithm \ref{QONS} with parameters $\eta=\min\left\{\frac{1}{8GD},\frac{\alpha}{2}\right\}$, $\epsilon=\frac{1}{\eta^2 D^2}$, $\left\{r_t=\frac{\rho G}{\pi n^3 (n/\rho +1) \beta t}\right\}_{t=1}^T$, can achieve the regret bound $O\left(\left(DG+\frac{1}{\alpha}\right)n\log(T+\log T)\right)$ with probability greater than $1-T\rho$, and its query complexity is $O(1)$ in each round.
\end{restatable}
The proof is given in Appendix \ref{proof:TNQ}.
    Replacing $\{r_t\}$ into $b$, we have $b=\log_2 \frac{n(n+\rho)t}{4\rho}=O(\log(Tn/\delta))$, and $c=\log_2{\frac{16\pi n}{\rho }}-1 = O(\log(Tn/\delta))$, where $\delta=T\rho$ is the total failure probability we set for the algorithm. 
Thus, $O(n\log(Tn/\delta))$ qubits are needed totally.

\section{Extensions}
We present how to extend the quantum online quasi-Newton methods to quantum adaptive gradient method and quantum online Newton method in this section.
\subsection{Quantum Adaptive Gradient Method}
\label{lab:SecExtension}
The adaptive gradient method (AdaGrad) proposed by Duchi et. al. \yrcite{duchi2011adaptive} is widely applied for deep neural networks. It has similar structure as online Newton methods, hence we adapt it as an extension of this work. 
The adaptive gradient method has nearly the same update rule as online Newton methods, except that the matrix in quasi-Newton direction is replaced by $G_t=\left(A_t\right)^{1/2}$. 
Note that it is computationally impractical to compute the root of the the outer product matrix $G_t$, Duchi et. al. \yrcite{duchi2011adaptive} specialized $G_t$ to a diagonal matrix $\text{diag}\left(A_t\right)^{1/2}$.
Then, both the inverse and root of the diagonal matrix can be computed in linear time so that the AdaGrad method has the same computational complexity as the first order of gradient method, and thus can be regarded as a modification of gradient descent.
It can also be generalized to $G_t=\text{diag}\left(A_t\right)^{1-p}$ for $p\in [0.5,1]$, which includes both the AdaGrad method and gradient descent method.

For those cases where only the zeroth-order oracle is accessible, we can use the $O(1)$-query quantum gradient estimation method to accelerate the calculation and keep the regret bound inline with the full information setting, which gives the quantum adaptive gradient method as shown in Algorithm \ref{QAGM}.

\begin{algorithm}[H]
    \caption{Quantum Adaptive Gradient Method}
    \label{QAGM}
    \begin{algorithmic}[1]
        \REQUIRE Step sizes $\gamma$, parameters $p\in [0.5,1], \epsilon, \{r_t\}$.
        \ENSURE $x_1, x_2, x_3, \dots x_T$
      \STATE  Choose the initial point $x_1\in\mathcal{K}$ randomly, let $A_0=\epsilon I_n$.
        \FOR{$t=1$ to $T$}
            \STATE Play $x_t$,  get the oracle of loss function $Q_{f_t}$ from the adversary.
            \STATE Use the quantum gradient estimation circuit with the parameter $r_t$ to get the estimated gradient $\widetilde{\nabla} f_t(x_t)$.
            \STATE Let $A_t=A_{t-1}+\widetilde{\nabla} f_t(x_t)(\widetilde{\nabla} f_t(x_t))^{\mathrm{T}} $.
            \STATE Update $x_{t+1}=\hat{P}^{(G_t)}_{\mathcal{K}}\left(x_t-\frac{1}{\gamma} G_t^{-1} \widetilde{\nabla} f_t(x_t)\right)$, where $G_t=\text{diag}\left(A_t\right)^{1-p}$.
        \ENDFOR
    \end{algorithmic}
\end{algorithm}

Although the method is called the gradient method, its analysis is closely related to online Newton methods. 
With similar analysis, appropriate parameters can be chosen to guarantee the following regret:
\begin{align}
    & \sum_{t=1}^T \left( f_t(x_t)-f_t(x^*) \right) \nonumber \\ 
    \leq & \frac{1}{2p\gamma} \trace{\left(G_T^p\right)} + \frac{D_q^2\gamma}{2} \trace{\left(G_T^p\right)}^{(1-p)/p} \nonumber \\ 
    \leq & O\left(D_q\sqrt{T}\right),
\end{align}
where $D_q$ is the upper bound of the $\ell_q$-diameter of the decision set, and $q=2p/(2p-1)$.
Note that this algorithm does not require the assumption of exponential concavity.

\subsection{Quantum Online Newton Method with Hessian Update}
\label{lab:SecON}

Since the estimation of each row of a Hessian matrix can be seen as a gradient estimation process of one dimension of the gradient function, the quantum Hessian estimation algorithm with the first-order oracle can be designed using the same framework of the quantum gradient estimation algorithm with the zeroth-order oracle. 
It outputs one row of the Hessian matrix per execution, thus $O(n)$ repeats of it can output a error controllable estimated Hessian matrix.
Here, with the quantum first-order oracle, given the superposition input, one can receive a superposition state of corresponding gradients at different points.
Figure \ref{QHEC} gives the quantum circuit of estimating the $i$-th row of the Hessian matrix.
Running such circuit for all $i\in [n]$, we will have a estimation of the Hessian.

\begin{figure*}[ht]
\begin{align*}
    \Qcircuit @C=1.8em @R=1.3em {
        \lstick{} & \gate{H^{\otimes b}} & \multigate{9}{ \ Q_{\nabla F_t} \ } & \qw & \qw                         & \qw & \multigate{9}{ \ Q_{\nabla F_t}^{-1} \ } & \gate{QFT^{-1}} & \meter & \cw\\ & \vdots & & & & & & \vdots & \vdots
        \\
        \lstick{} & \gate{H^{\otimes b}} & \ghost{ \ Q_{\nabla F_t} \ } & \qw       & \qw                         & \qw & \ghost{ \ Q_{\nabla F_t}^{-1} \ }        & \gate{QFT^{-1}} & \meter & \cw \\
        \lstick{} & \gate{H^{\otimes b}} & \ghost{ \ Q_{\nabla F_t} \ } & \qw       & \qw                         & \qw & \ghost{ \ Q_{\nabla F_t}^{-1} \ }        & \gate{QFT^{-1}} & \meter & \cw \\
        \lstick{} & \gate{H^{\otimes b}} & \ghost{ \ Q_{\nabla F_t} \ } & \qw       & \qw                         & \qw & \ghost{ \ Q_{\nabla F_t}^{-1} \ }        & \gate{QFT^{-1}} & \meter & \cw
            \inputgroupv{1}{5}{.1em}{4.9em}{\ket{0^{n\otimes b}}} \\
        \lstick{\ket{0^{\otimes c}}_1} & \qw & \ghost{ \ Q_{\nabla F_t} \ } & \qw    & \qw & \qw & \ghost{ \ Q_{\nabla F_t}^{-1} \ }        & \qw  \\ & \vdots & & & & & & \vdots & 
        \\
        \lstick{\ket{0^{\otimes c}}_i} & \qw & \ghost{ \ Q_{\nabla F_t} \ } & \qswap    & \qw & \qswap & \ghost{ \ Q_{\nabla F_t}^{-1} \ }        & \qw  \\ & \vdots & & \qwx & & \qwx & & \vdots & 
        \\
        \lstick{\ket{0^{\otimes c}}_n} & \qw & \ghost{ \ Q_{\nabla F_t} \ } & \qswap  \qwx    & \multigate{1}{\ + (\mod 2^c) \ } & \qswap \qwx & \ghost{ \ Q_{\nabla F_t}^{-1} \ }        & \qw  \\
        \lstick{\ket{y_0}} & \qw  & \qw      & \qw                      & \ghost{\ + (\mod 2^c) \ }       & \qw                         & \qw  & \qw  
    }
\end{align*}
\caption{Quantum circuit of estimating the $i$-th row of the Hessian matrix.  \label{QHEC}}
\end{figure*}
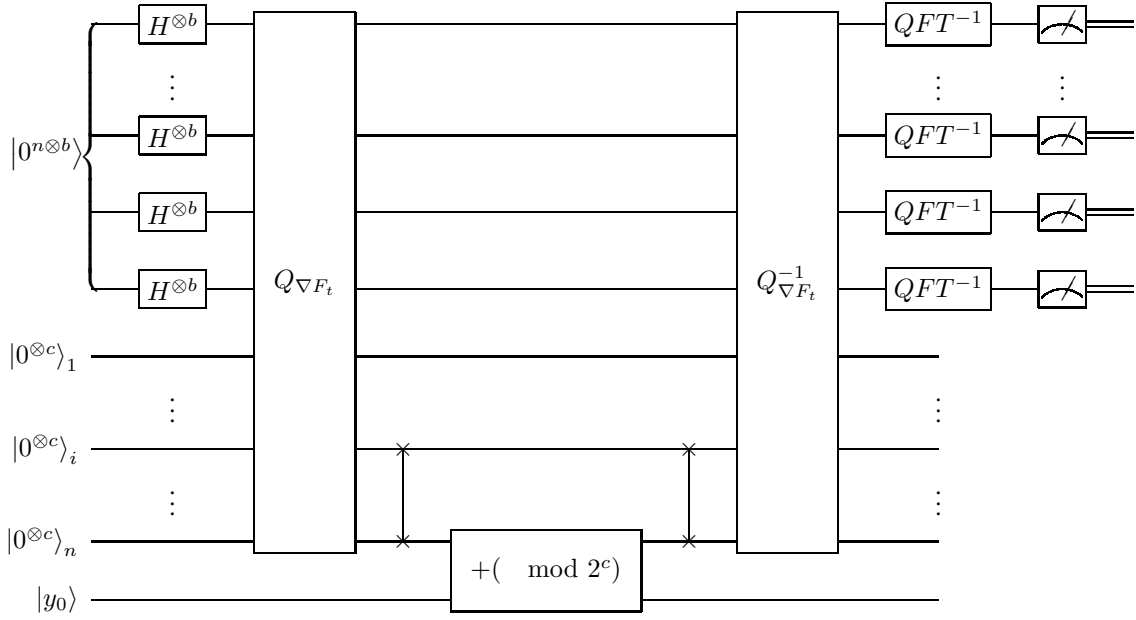

Compared with the quantum gradient estimation circuit, the sub-circuits $Q_{F_t}$ and $Q_{F_t}^{-1}$ are replaced by the first-order version $Q_{\nabla F_t}$ and $Q_{\nabla F_t}^{-1}$, respectively, where $\nabla F_t(u)=\cfrac{2^b}{2Lr_t} \left [ \nabla f_t \left (x_t+\cfrac{r_t}{2^b} \left (u-\cfrac{2^b}{2}\mathbbm{1} \right ) \right )- \nabla f_t(x_t) \right ]$, which needs $n$ $c$-qubit registers to store the output. 
Here, $L$ is the Lipschitz constant of the gradient function, and the other symbols have the same meaning as those in the quantum gradient estimation method.
Before and after the modular addition, swap gates between the $i$-th $c$ qubits register and the last $c$ qubits register are added.
Then, the estimation of the element in $i$-th row and $j$-th column of the Hessian matrix is given by $\widetilde{H}_t^{(i,j)}=\cfrac{2L}{2^b} \left(m_j-\cfrac{2^b}{2}\right)$, where $m_j$ is the measurement result of the $j$-th $b$-qubit register, for all $j\in [n]$.

Noticing that there is no online algorithm which uses the Hessian information to design the update rule.
Thus, we also give a classical framework online Newton method with Hessian update, and show that it can achieve logarithm regret.  
We still study the online exp-concave optimization problem, but assume here the player can get the gradient and the Hessian of the loss function in each round efficiently after suffering the loss.
Consequentially, the loss functions should assume to be twice differentiable. 
We assume that the gradient of the loss functions is $L$-Lipschitz continuous, that is, $|\nabla f_t(x)-\nabla f_t(y)| \leqslant L \|y-x\|,\quad \forall x, y \in \mathcal{K}$.
Similar to the exp-concave assumption, we further assume that $\forall x,y \in \mathcal{K}, f_t(y) \geq f_t(x) + \nabla f_t(x)^{\mathrm{T}}(y-x)+\frac{\eta}{2}(y-x)^{\mathrm{T}} H(f_t)(x) (y-x)$, for all $t\in [T]$, where $H(f_t)(x)$ is the Hessian of the loss function $f_t$ at point $x$. 
The right-hand side of this assumption is the Taylor expansion of the loss function omitting high-order terms and multiplying a factor to the second-order term, and thus will be easy to satisfy when $\eta$ is small since the Hessian of a convex function is positive semi-definite.

\begin{algorithm}[H]
    \caption{Online Newton Method with Hessian Update}
    \label{ONSH}
    \begin{algorithmic}[1]
        \REQUIRE Step sizes $\eta$, parameters $\epsilon$.
        \ENSURE $x_1, x_2, x_3, \dots x_T$
      \STATE  Choose the initial point $x_1\in\mathcal{K}$ randomly, let $A_0=\epsilon I_n$.
        \FOR{$t=1$ to $T$}
            \STATE Play $x_t$, observe the loss function $f_t$ from the adversary.
            \STATE Update $x_{t+1}=\hat{P}^{(A_t)}_{\mathcal{K}}\left(x_t-\frac{1}{\eta} A_t^{-1} \nabla f_t(x_t)\right)$, where $A_t=A_{t-1}+H_t$.
        \ENDFOR
    \end{algorithmic}
\end{algorithm}

We set $x_{t+1}=\hat{P}^{(A_t)}_{\mathcal{K}}\left(x_t-\frac{1}{\eta} A_t^{-1} \nabla f_t(x_t)\right)$ as the update rule, where $A_t=A_{t-1}+H(f_t)(x_t)$. Below we will write the Hessian $H(f_t)(x_t)$ as $H_t$ for short. The algorithm is presented in Algorithm \ref{ONSH}, and then we show how to choose appropriate parameters such that Algorithm \ref{ONSH} guarantees logarithm regret, which gives Theorem \ref{TNHQ}.

\begin{restatable}{mythe}{TNHQ}
    \label{TNHQ}
    Algorithm \ref{ONSH} with parameters $\eta=\frac{1}{2DL}$ and $\epsilon=L$, can achieve the regret bound $O\left(\frac{DLn}{\alpha}\log(T+1)\right)$.
\end{restatable}
The proof is given Appendix in \ref{proof:TNHQ}.

\subsection{Extend to non-smooth loss functions}
\label{lab:SecNonSmooth}
In this section, we present how to extend Algorithm \ref{QONS} to non-smooth loss functions. 
It has been shown that Lipschitz continue functions are smooth in a small region with high probability \cite{Childs18}.
However, adopting this technique necessitates a few additional steps. 
In their proof, they shows that non-smooth loss functions are still smooth in a small
region with high probability by bounding the trace of the Hessian matrices.
As they do not assume that the loss functions are second-order differentiable (nor do we), they need to use the mollifier functions to approximate the loss functions.
Yet, for certain loss functions, the approximation error becomes uncontrollable.
To address this challenge, they introduce classical randomness into the gradient evaluation, that is, sampling randomly from the neighborhood of any given point, and then use the gradient of the sample point as the gradient of the given point.
We denote the sample point as $z_t\in B_{\infty}(x_t,r^{\prime}_t)$, where $B_{\infty}(x_t,r^{\prime}_t)$ is the ball in $L_{\infty}$ norm with radius $r^{\prime}_t\in \mathbb{R}$ and center $x_t$. 
Lemma \ref{Smooth} shows that for any sample point $z_t$, the trace of the Hessian matrices is small in a region with high probability, and therefore we can let $\beta=nG/p_t r^{\prime}_t$.

\begin{restatable}{mylem}{Smooth}
\label{Smooth}
    {(\cite{Childs18})}
    In Algorithm \ref{QONS}, for all timestep $t$, let $f_t:\mathcal{K} \to \mathbb{R}$ be a convex function with Lipschitz parameter $G$. Then for any $r_t,r^{\prime}_t>0$ and $1 \geq p_t>0$, we have
    \begin{align}
        \Pr_{z_t\in B_{\infty}(x_t,r^{\prime}_t)} \left[\exists y\in B_{\infty}(z_t,r_t), \Tr{\nabla^2 f_t(y)} \geq \frac{nG}{p_t r^{\prime}_t}\right] 
        \leq  p_t.
    \end{align}
\end{restatable}

By combining the observation that in the proof of the quantum gradient estimation, the smooth condition is only used in a small region, we can give a process (for each round) to extend our algorithm to the non-smooth case, as follows:
\begin{itemize}
     \item[1] Sample $z_t\in B_{\infty}(x_t,r^{\prime}_t)$.
     \item[2] Apply the quantum gradient estimation circuit on point $z_t$ to get the estimated gradient $\widetilde{\nabla} f_t(z_t)$, let $\widetilde{\nabla} f_t(x_t)=\widetilde{\nabla} f_t(z_t)$
     \item[3] Update $x_{t+1}=\hat{P}^{(A_t)}_{\mathcal{K}}\left(x_t-\frac{1}{\eta_t} A_t^{-1} \widetilde{\nabla} f_t(x_t)\right)$, where $A_t=A_{t-1}+\widetilde{\nabla} f_t(x_t)(\widetilde{\nabla} f_t(x_t))^{\mathrm{T}}$.
\end{itemize}
Then, by Lemma \ref{QGB}, we will establish the error bound of the form $\norm{\nabla f_t(z_t) - \widetilde{\nabla} f_t(z_t)}_1$. 
However, Lemma \ref{QNSGB} is no longer applicable here as it requires the error bound of $\norm{\nabla f_t(x_t) - \widetilde{\nabla} f_t(x_t)}_1$. 
To address this, we give Lemma \ref{QNSGBs}. Initially, we utilize the property of exp-concavity at the point $z_t$, which yields a relationship between $y$ and $z_t$. Subsequently, we leverage this relationship to establish a connection between $y$ and $x_t$.

\begin{restatable}{mylem}{QNSGBs}
    \label{QNSGBs}
    In the non-smooth version of Algorithm \ref{QONS}, for all timestep $t$, let $f_t:\mathcal{K} \to \mathbb{R}$ be $\alpha$-exp-concave functions with Lipschitz parameter $G$, where $\mathcal{K}$ is a convex set with diameter $D$, then for any $y\in \mathcal{K}$ and any $\eta_t\leq \min \left\{\frac{1}{8GD},\frac{\alpha}{2}\right\}$, the estimated gradient $\widetilde{\nabla} f_t(x_t)$ satisfies
    \begin{align}
        \label{QNSGB2}
        f_t(y)\geq & f_t(x_t)+\widetilde{\nabla} f_t(x_t)^{\mathrm{T}}(y-x_t) \nonumber \\ 
        &+\frac{\eta_t}{2}(y-x_t)^{\mathrm{T}} \widetilde{\nabla} f_t(x_t) \widetilde{\nabla} f_t(x_t)^{\mathrm{T}}(y-x_t) \nonumber \\
        &- \frac{9D}{8} \norm{\nabla f_t(z_t) - \widetilde{\nabla} f_t(z_t)}_1 \nonumber \\
        &- \frac{D}{16G} \norm{\nabla f_t(z_t) - \widetilde{\nabla} f_t(z_t)}^2_1 \nonumber \\
        & -3G\sqrt{n}r_t^{\prime} - \frac{ n r_t^{\prime 2} G}{2}.
    \end{align}
\end{restatable}

The proof is given in Appendix \ref{proof:QNSGBs}. The remaining part involves parameters selections, which is similar to the proof of Theorem \ref{TNQ}.

\begin{restatable}{mythe}{TNQs}
    \label{TNQs}
    The non-smooth version of Algorithm \ref{QONS} with parameters $\eta=\min\left\{\frac{1}{8GD},\frac{\alpha}{2}\right\}$, $\epsilon=\frac{1}{\eta^2 D^2}$, $\left\{r_t=\frac{\rho p}{\pi n^{9/2} (n/\rho +1)  t^2}\right\}_{t=1}^T$, $\left\{r^{\prime}_t=\frac{1}{t\sqrt{n}}\right\}_{t=1}^T$ can achieve the regret bound $O\left(\left(DG+\frac{1}{\alpha}\right)n\log(T+\log T)\right)$ with probability greater than $1-T(\rho+p)$, and its query complexity is $O(1)$ in each round.
\end{restatable}

The proof is given in Appendix \ref{proof:TNQs}.

\section{Conclusion and Discussion}
\label{lab:SecConclusion}

\begin{table*}[ht]
    \centering
    \begin{threeparttable}
    \caption{Summary of the update rules.}
    \begin{tabular}{l l}
        \toprule
         Algorithm & Update rules \\
        \midrule
        \makecell[l]{Quantum online quasi-Newton method \\ (Section \ref{lab:SecQON})} & \makecell[l]{$x_{t+1}=\hat{P}^{(A_t)}_{\mathcal{K}}\left(x_t-\frac{1}{\eta} A_t^{-1} \widetilde{\nabla} f_t(x_t)\right)$, \\ where $A_t=A_{t-1}+\widetilde{\nabla} f_t(x_t)(\widetilde{\nabla} f_t(x_t))^{\mathrm{T}}$} \\
        \makecell[l]{Quantum adaptive gradient method \\ (Subsection \ref{lab:SecExtension})} & \makecell[l]{$x_{t+1}=\hat{P}^{(G_t)}_{\mathcal{K}}\left(x_t-\frac{1}{\gamma} G_t^{-1} \widetilde{\nabla} f_t(x_t)\right)$, \\ where $G_t=\text{diag}\left(A_t\right)^{1-p}$ and $A_t=A_{t-1}+\widetilde{\nabla} f_t(x_t)(\widetilde{\nabla} f_t(x_t))^{\mathrm{T}}$} \\
        \makecell[l]{Quantum online Newton method \\ (Subsection \ref{lab:SecON})} & \makecell[l]{$x_{t+1}=\hat{P}^{(A_t)}_{\mathcal{K}}\left(x_t-\frac{1}{\eta} A_t^{-1} \nabla f_t(x_t)\right)$, \\ where $A_t=A_{t-1}+\widetilde{H}_t$}\\
        \bottomrule
    \end{tabular}
    \label{tab:SU}
    \end{threeparttable}
\end{table*}

In this paper, we considered the multi-points bandit feedback setting for online exp-concave optimization against the completely adaptive adversary. 
We provided a quantum algorithm and proved that it can achieve $O(n\log{T})$ regret where only $O(1)$ queries were needed in each round.
These results showed that potential quantum advantages are possible for problems of online exp-concave optimization under such setting since the classical online Newton method with zero-order feedback only achieves $O(n T^{2/3}\log T)$ regret.
Furthermore, we presented how to extend our work to quantum adaptive gradient method and quantum online Newton method with Hessian update.
For the Hessian update case, we also showed how to extend the quantum gradient estimation algorithm to the quantum Hessian estimation algorithm for the case that the quantum first-order oracle is accessible.
Table \ref{tab:SU} summarizes the update rules of these algorithms.

We also proposed a classical framework of online Newton method with Hessian update, and showed that it can achieve logarithm regret with appropriate parameters.
There are still much room for research between the Hessian update (also known as full-rank update) and the rank-1 update as shown in Algorithm \ref{ONS} and \ref{QONS}, namely, online Newton methods with rank-$k$ update for $1<k<n$ have not been explored yet.
In the offline optimization setting, the convergence rate is better with higher rank update \cite{liu2023symmetric,liu2023block}, but for the online optimization setting, the problem whether high rank update can improve the regret bound still remains open.

Furthermore, online algorithms with the quantum first-order oracle have not been fully explored.
With the quantum first-order oracle, given the superposition input, one can receive a superposition state of corresponding gradients at different points, but due to the properties of quantum computing, we need to explore algorithmic design to utilize some kind of overall information if we want to keep the advantage of quantum computing.
Thus, on one hand, one needs to answer whether it can guarantee a better regret if we can make use of some kind of overall information of gradients of the loss function at different points.
On the other hand, one also needs to consider how to design a quantum algorithm to estimate the Hessian matrix.
But if we design the quantum Hessian estimation algorithm with the first-order oracle using the same framework of the quantum gradient estimation algorithm with the zeroth-order oracle, it can only output one row of the Hessian matrix per execution, which will lead to $O(n)$ queries for estimating the total Hessian matrix.
Thus, it's also interesting to consider whether there exists a quantum algorithm that can estimate the Hessian with $O(1)$-query to the quantum first-order oracle.

\section*{Acknowledgements}
The work of John C.S. Lui was supported in part by the RGC SRFS2122-4S02. The work of Lvzhou Li was supported in part by the National Science Foundation of China under Grant 62272492 and the Guangdong Provincial Quantum Science Strategic Initiative under Grant GDZX2303007.

\section*{Impact Statement} 
This paper presents work whose goal is to advance the field of Machine Learning. There are many potential societal consequences of our work, none which we feel must be specifically highlighted here.

\normalem
\bibliography{QONS.bib}
\bibliographystyle{icml2024}

\newpage
\appendix
\onecolumn

\section{Extended Related Works}
\label{lab:SecERW}
\subsection{Extended Related Works of Online Convex Optimizaiton}
\label{Subsec:ERWOCO}

In general, less feedback information received by the player in each round will lead to a worse regret. In early works \cite{Zinkevich03,Hazan2007}, it was assumed that the player can get full information of loss functions as feedback, or has access to gradient oracles of loss functions. 
In 2003, Zinkevich \yrcite{Zinkevich03} defined the online convex optimization model and showed that with gradient oracles feedback, the online gradient descent could achieve $O(\sqrt{T})$ regret.
Additionally, it has been shown that the lower bound of this setting is $\Omega(\sqrt{T})$ (Theorem 3.2 of \cite{Hazan16}).

Contrarily, online convex optimization in a single-query bandit setting was proposed by Flaxman et al. \yrcite{Flaxman2005}, where the only feedback was the value of the loss. 
Note that in the bandit setting, a regret bound for any strategy against the completely adaptive adversary is $\Omega(T)$. 
Thus, it needs to assume that the adversary is either adaptive or oblivious, i.e., the adversary must choose the loss function before observing the player's choice or before the start of the game, respectively. 
In any case, the expected regret of Flaxman's algorithm is $O(\sqrt{n}T^{3/4})$. 
The dependence on $T$ was improved to $O(\text{poly}(n)\sqrt{T}\log{T})$ by Bubeck and Eldan \yrcite{Bubeck2016,bubeck2017} and Lattimore \yrcite{lattimore2020} with the price of increasing the dimension-dependence. 
    
Better regret can be achieved if the player can query the value of the loss function at more than one point in each round. 
In 2010, Agarwal et al. \yrcite{Agarwal2010} considered the multi-query bandit setting, and proposed an algorithm with an expected regret bound of $O(n^2\sqrt{T})$, where the player queries $O(1)$ points in each round. 
In 2017, the upper bound was improved to $O(\sqrt{nT/k})$ by Shamir \yrcite{Shamir2017}, where the player queries $k$ points in each round. 
Although much effort has been made to minimize the impact of the dimension, but there is still a polynomial dependence on it, until the introducing of quantum techniques. In 2022, He et al. \yrcite{he2022quantum} gave a quantum subgradient descent method that achieves $O(\sqrt{T})$ regret without dependence of the dimension $n$. 

For works of the general convex loss functions setting mentioned above, it seems difficult to achieve regret logarithmic in $T$ even with the help of quantum computers. 
But for specific loss functions, achieving logarithmic regret is possible. For strongly convex loss functions setting, $O(\log{T})$ regret could be achieved by using online gradient descent with full information feedback \cite{Hazan2007} and by using quantum online subgradient descent with zeroth-order feedback \cite{he2022quantum}. 
For exp-concave loss functions setting, Hazan et al. \yrcite{Hazan2007} showed that $O(n\log{T})$ regret could be achieved by using online Newton step method with full information feedback. 
The online Newton step method requires a generalized projection. For the full information feedback or first-order feedback setting, such projection can be avoided by using barrier regularizer for some cases \cite{abernethy2012interior,mhammedi2022damped,mhammedi2023quasi}.
But the setting of exp-concave loss functions with zeroth-order feedback is still an open problem.
To the best of our best knowledge, the classical online Newton method with zeroth-order feedback only achieves $O(nT^{2/3}\log T)$ regret\cite{liu2018online}. 

Table \ref{tab:RB} provides a comprehensive summary for various works in online convex optimization.

\begin{table}[htbp]
    \centering
    \begin{threeparttable}
    \caption{Regret bound for online convex optimization with different settings.}
    \begin{tabular}{l|l|c|c|c}
        \hline
         & Paper & Feedback & Adversary & Regret \\
        \hline
        \multicolumn{5}{c}{general convex loss functions} \\
        \hline
        \multirow{2}{*}{1} & \cite{Zinkevich03} & Full information (first-order oracles) & Completely & $O\left(DG\sqrt{T}\right)$ \\
        & \cite{Hazan16} & Full information (first-order oracles) & Completely  & $\Omega\left(DG\sqrt{T}\right)$ \\
        \hline
        \multirow{4}{*}{2} & \cite{Flaxman2005} & Single query (zeroth-order oracles) & Oblivious & $O\left(\sqrt{DGCn}T^{3/4}\right)$ \\
        & \cite{Bubeck2016} & Single query (zeroth-order oracles) & Oblivious & $O\left(DGn^{11} \sqrt{T}\log T \right)$ \\
        & \cite{bubeck2017} & Single query (zeroth-order oracles) & Oblivious & $O\left(DGn^{9.5} \sqrt{T}\log T \right)$ \\
        & \cite{lattimore2020} & Single query (zeroth-order oracles) & Oblivious & $O\left(DGn^{2.5} \sqrt{T}\log T \right)$ \\
        \hline
        \multirow{4}{*}{3} & \cite{Agarwal2010} & $O(1)$ queries (zeroth-order oracles) & Adaptive & $O\left((D^2+n^2G^2)\sqrt{T}\right)$ \\
        & \cite{Shamir2017} & $k$ queries (zeroth-order oracles) & Adaptive  & $O\left(DG\sqrt{nT/k}\right)$ \\
        & \cite{he2022quantum} & $O(n)$ queries (zeroth-order oracles) & Completely & $O\left(DG\sqrt{T}\right)$ \\
        & \cite{he2022quantum} & $O(1)$ queries (quantum zeroth-order oracles) & Completely & $O\left(DG\sqrt{T}\right)$ \\
        \hline
        \multicolumn{5}{c}{ $\alpha$-strongly convex function } \\
        \hline
        \multirow{3}{*}{4} & \cite{Hazan2007} & Full information (first-order oracles) & Completely & $O\left(G^2 \log T\right)$ \\
        & \cite{Shamir2013} & Single query (zeroth-order oracles) & Completely & $\Omega\left(n \sqrt{T}\right)$ \\
        & \cite{he2022quantum} & $O(1)$ queries (quantum zeroth-order oracles) & Completely & $O\left(G^2 \log T\right)$ \\
        \hline
        \multicolumn{5}{c}{ $\alpha$-exp-concave functions } \\
        \hline
        \multirow{3}{*}{5} & \cite{Hazan2007} & Full information (first-order oracles) & Completely & $O\left(DGn \log T\right)$ \\
        & \cite{liu2018online} & Single queries (zeroth-order oracles) & Oblivious & $O\left(DGn T^{2/3} \log T\right)$ \\
        & This work & $O(1)$ queries (quantum zeroth-order oracles) & Completely & $O\left(DGn \log (T+\log T)\right)$\\
        \hline
    \end{tabular}
    \begin{tablenotes}
        \item [-] Regret bound for online convex optimization with different settings. For a strategy, the less feedback information that the player uses, the better; the stronger adversary the player faces, the better. The full information feedback model reveals the most information about the function, while the single-query feedback model reveals the least. The completely adaptive adversary is strongest, thus the corresponding models are the least restrictive ones in using, while the oblivious adversary is weakest, thus the corresponding models are the most stringent ones in using.
        \item[-] $D$ is the diameter upper bound of  the feasible set, $G$ Lipschitz constant, $C$ is the diameter upper bound of the range of loss functions, $n$ is the dimension of the feasible set, $T$ is the horizon.
    \end{tablenotes}
    \label{tab:ERB}
    \end{threeparttable}
\end{table}

\subsection{Extended Related Works of Quantum Optimization}
\label{Subsec:ERWQO}

Optimization problem is one of the most important problems in science and engineering.
Since quantum computing has shown its advantages over classical computing, e.g., the famous Shor algorithm \cite{Shor1999}, Grover algorithm \cite{Grover1996} and HHL algorithm \cite{Harrow2009,Harrow2017}, people have been exploring how to accelerate the optimization process using quantum computing techniques.
In recent years, some significant advantages have been made on quantum algorithms for optimization problems. We classify these quantum optimization works in four classes: whether the problems are continuous or discrete, and whether the problems are online or offline.

\noindent \textbf{Offline discrete optimization.} This kind of optimization problems are also called combinatorial optimization, which was shown to be acceleratable by using quantum techniques such as Grover's algorithm or quantum walks \cite{Grover1996,Ambainis2006,Durr2006,Durr1996,Mizel2009,Yoder2014,Sadowski2015,He2020}. In the NISQ area, series of variational quantum algorithms were developed to solve combinatorial optimization, readers are referred to \cite{blekos2023review} for more information.

\noindent \textbf{Offline continuous optimization.} When the optimization problems are continuous, usually one can device efficient algorithms because continuous problems tend to have more properties to exploit such as the gradient information.  
In recent years, some quantum improvements were achieved for offline continuous optimization in linear programming \cite{KP18,Li2019,Van2019}, second-order cone programming \cite{Kerenidis2019s,Kerenidis2019Svm,Kerenidis2019P}, quadratic programming \cite{kerenidis2020}, polynomial optimization \cite{Rebentrost2019}, semi-definite optimization \cite{KP18,Joran19,Fernando17,Fernando19,Joran17}, and general convex optimization \cite{DeWolf18,Childs18}. 
For the stochastic convex bandits problem, \cite{li2022quantum} gave a quantum algorithm that achieved exponential speedup in $T$ with the price of increasing the dimension-dependence.
Quantum techniques such as quantum Fourier transform, HHL-like algorithms and QRAM play an important role in offline continuous optimization.

\noindent \textbf{Online discrete optimization.} Recently, people are considering how to apply quantum computing methods to online optimization problems. 
The online discrete optimization problem is best known as the multi-arm bandit problem. 
In 2020, two quantum algorithms for the best arm identification problem, a central problem in multi-armed bandit, were given \cite{casale2020, wang2020}. 
They showed that with the help of quantum coherent bandit oracle, quantum speedup can be achieved in the best arm identification problem. 
Later, for multi-armed bandits and stochastic linear bandits problems (the later problem is continuous), exponentially improving for the dependence in $T$ was shown using a weaker quantum bandit oracle \cite{wan2023quantum}. 
Quantum amplitude amplification and quantum Monte Carlo method play a central role in the study of quantum multi-arm bandit.

\noindent \textbf{Online continuous optimization.} In online continuous optimization, it is usually assumed that there is an adversary who is responsible for choosing the loss functions in each round. Thus, it is often called the adversarial optimization problems. 
In 2022, \cite{he2022quantum} studied the quantum algorithm of online convex optimization, which achieves a regret bound that is independent of the dimension so that outperforms the optimal classical algorithms in the same setting which still have polynomial dependence of the dimension.
For the online portfolio optimization problem, \cite{lim2022quantum} gave a quantum algorithm that provides a quadratic speedup in the time complexity by using techniques such as quantum state preparation, inner product estimation and multi-sampling.

\section{Proof detail}
\label{lab:SecProof}

\subsection{Proof of Lemma \ref{QGB}}
\label{proof:QGE}

We first provide some basic definitions of quantum gates/circuits \cite{Nielsen2002}, as follows:
\begin{itemize}
    \item Hadamard gate: $H\ket{0}:=\frac{\ket{0}+\ket{1}}{\sqrt{2}}$, $H\ket{1}:=\frac{\ket{0}-\ket{1}}{\sqrt{2}}$,
    \item  modulo add: $+\ket{u}\ket{v}:=\ket{u}\ket{(u+v)\mod 2^c}$,
    \item Inverse quantum Fourier transform: $QFT^{-1}\ket{v}:=\frac{1}{2^b}\sum_{v\in\{-2^{b-1},\dots,2^{b-1}\}}e^{-\frac{2\pi i vw}{2^b}}\ket{w}, \forall w \in \{-2^{b-1},\dots,2^{b-1}\}$.
\end{itemize}

The oracle $Q_{F_t}$, defined as $Q_{F_t}\ket{u}\ket{0}:=\ket{u}\ket{0+F_t(u)}$, where $F_t(u)=\cfrac{2^b}{2Gr_t} \left [f_t \left (x_t+\cfrac{r_t}{2^b} \left (u-\cfrac{2^b}{2}\mathbbm{1} \right ) \right )-f_t(x_t) \right ]$, can be construct by two $Q_{f_t}$, as follows:
\begin{align}
    & \ket{u}\ket{x_t}\ket{0}\ket{0}\ket{0}\ket{0}\ket{0} \nonumber \\
    \longrightarrow &  \ket{u}\ket{x_t}\ket{x_t+\cfrac{r_t}{2^b} \left (u-\cfrac{2^b}{2}\mathbbm{1} \right )}\ket{0}\ket{x_t}\ket{0}\ket{0} \nonumber \\
    \stackrel{\text{two } Q_{f_t}}\longrightarrow &  \ket{u}\ket{x_t}\ket{x_t+\cfrac{r_t}{2^b} \left (u-\cfrac{2^b}{2}\mathbbm{1} \right )}\ket{f_t \left (x_t+\cfrac{r_t}{2^b} \left (u-\cfrac{2^b}{2}\mathbbm{1} \right ) \right )}\ket{x_t}\ket{f_t(x_t)}\ket{0} \nonumber \\
    \longrightarrow &  \ket{u}\ket{x_t}\ket{x_t+\cfrac{r_t}{2^b} \left (u-\cfrac{2^b}{2}\mathbbm{1} \right )}\ket{f_t \left (x_t+\cfrac{r_t}{2^b} \left (u-\cfrac{2^b}{2}\mathbbm{1} \right ) \right )}\ket{x_t}\ket{f_t(x_t)}\ket{F_t(u)},
\end{align}
and then uncompute all of the auxiliary registers.

Now we are ready to prove Lemma \ref{QGB}.

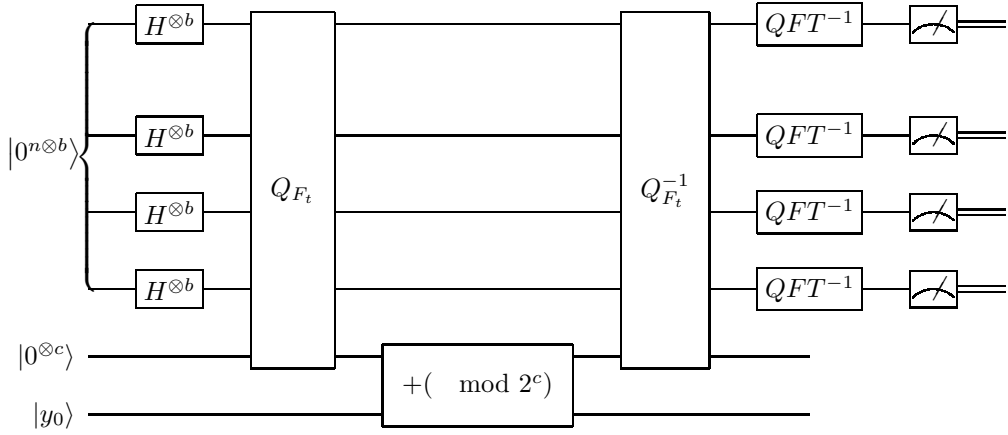
\begin{figure*}[htp]
\begin{align*}
    \Qcircuit @C=1.8em @R=1.3em {
        \lstick{} & \gate{H^{\otimes b}} & \multigate{5}{ \ Q_{F_t} \ } & \qw                          & \multigate{5}{ \ Q_{F_t}^{-1} \ } & \gate{QFT^{-1}} & \meter & \cw\\
        \\
        \lstick{} & \gate{H^{\otimes b}} & \ghost{ \ Q_{F_t} \ }        & \qw                          & \ghost{ \ Q_{F_t}^{-1} \ }        & \gate{QFT^{-1}} & \meter & \cw \\
        \lstick{} & \gate{H^{\otimes b}} & \ghost{ \ Q_{F_t} \ }        & \qw                          & \ghost{ \ Q_{F_t}^{-1} \ }        & \gate{QFT^{-1}} & \meter & \cw \\
        \lstick{} & \gate{H^{\otimes b}} & \ghost{ \ Q_{F_t} \ }        & \qw                          & \ghost{ \ Q_{F_t}^{-1} \ }        & \gate{QFT^{-1}} & \meter & \cw
            \inputgroupv{1}{5}{.1em}{4.9em}{\ket{0^{n\otimes b}}} \\
        \lstick{\ket{0^{\otimes c}}} & \qw & \ghost{ \ Q_{F_t} \ }      & \multigate{1}{\ + (\mod 2^c) \ }& \ghost{ \ Q_{F_t}^{-1} \ }        & \qw  \\
        \lstick{\ket{y_0}} & \qw         & \qw                      & \ghost{\ + (\mod 2^c) \ }       & \qw                           & \qw  
    }
\end{align*}
\caption{Quantum gradient estimation.  \label{QGEC}}
\end{figure*}

\QGB*

\begin{proof}
    We omit the subscript $t$ in the proof as the lemma holds for each timestep $t$. 
    The states after Step $4$ will be:
    \begin{align}
        \frac{1}{\sqrt{2^n}}\sum_{a\in\{0,1,\dots,2^n-1\}}e^{\frac{2\pi i a}{2^n}} \ket{0^{\otimes b},0^{\otimes b},\dots,0^{\otimes b}} \ket{0^{\otimes c}} \ket{a}.
    \end{align}
    After Step $5$:
    \begin{align}
       \frac{1}{\sqrt{2^{bn+c}}}\sum_{u_1,u_2,\dots,u_n\in\{0,1,\dots,2^b-1\}} 
       \sum_{a\in\{0,1,\dots,2^c-1\}} 
       e^{\frac{2\pi i a}{2^n}} 
       \ket{u_1,u_2,\dots,u_n} 
       \ket{0^{\otimes c}} \ket{a}. 
    \end{align}
    After Step $6$:
    \begin{align}
       \frac{1}{\sqrt{2^{bn+c}}}\sum_{u_1,u_2,\dots,u_n\in\{0,1,\dots,2^b-1\}} 
       \sum_{a\in\{0,1,\dots,2^c-1\}} 
       e^{\frac{2\pi i a}{2^n}}   
       \ket{u_1,u_2,\dots,u_n} 
       \ket{F(u)} \ket{a}. 
    \end{align}
    After Step $7$:
    \begin{align}
       \frac{1}{\sqrt{2^{bn+c}}}\sum_{u_1,u_2,\dots,u_n\in\{0,1,\dots,2^b-1\}} 
       \sum_{a\in\{0,1,\dots,2^c-1\}} 
       e^{2\pi i F(u)}
       e^{\frac{2\pi i a}{2^n}} 
       \ket{u_1,u_2,\dots,u_n} 
       \ket{F(u)} \ket{a}. 
    \end{align}
    After Step $8$:
    \begin{align}
       \frac{1}{\sqrt{2^{bn+c}}}\sum_{u_1,u_2,\dots,u_n\in\{0,1,\dots,2^b-1\}} 
       \sum_{a\in\{0,1,\dots,2^c-1\}} 
       e^{2\pi i F(u)}
       e^{\frac{2\pi i a}{2^n}} 
       \ket{u_1,u_2,\dots,u_n} 
       \ket{0^{\otimes c}} \ket{a}. 
    \end{align}
    In the following, the last two registers will be omitted:
    \begin{align}
       \frac{1}{\sqrt{2^{bn}}}\sum_{u_1,u_2,\dots,u_n\in\{0,1,\dots,2^b-1\}} 
       e^{2\pi i F(u)}
       \ket{u_1,u_2,\dots,u_n} .
    \end{align}
    And then we simply relabel the state by changing $u \to v=u-\frac{2^b}{2}$:
    \begin{align}
        \label{phi}
        \frac{1}{\sqrt{2^{bn}}}\sum_{v_1,v_2,\dots,v_n\in\{-2^{b-1},-2^{b-1}+1,\dots,2^{b-1}\}} 
        e^{2\pi i F(v)}
        \ket{v}. 
    \end{align}
    We denote Formula (\ref{phi}) as $\ket{\phi}$. Let $g=\nabla f(x)$, and consider the idealized state
    \begin{align}
        \label{psi}
        \ket{\psi}=\frac{1}{\sqrt{2^{bn}}}\sum_{v_1,v_2,\dots,v_n\in\{-2^{b-1},-2^{b-1}+1,\dots,2^{b-1}\}} 
        e^{\frac{2\pi i g \cdot v}{2G}}
        \ket{v}.
    \end{align}
    After Step $10$, from the analysis of phase estimation \cite{Brassard00}:
    \begin{align}
        \Pr[\abs{\frac{Ng_i}{2G}-m_i}>e]<\frac{1}{2(e-1)}, \forall{i\in [n]}.
    \end{align}
    Let $e=n/\rho +1$, where $1\geq \rho>0$. We have
    \begin{align}
        \Pr[\abs{\frac{Ng_i}{2G}-m_i}>n/\rho +1]<\frac{\rho}{2n}, \forall{i\in [n]}.
    \end{align}
    Note that the difference in the probabilities of measurement on $\ket{\phi}$ and $\ket{\psi}$ can be bounded by the trace distance between the two density matrices:
    \begin{align}
        \|\dyad{\phi}{\phi}-\dyad{\psi}{\psi}\|_1
        =2\sqrt{1-|\braket{\phi}{\psi}|^2}
        \leq 2\|\ket{\phi}-\ket{\psi}\|.
    \end{align}
    Since $f$ is $\beta$-smooth, we have
    \begin{align}
        F(v) & \leq \frac{2^b}{2Gr}[f(x+\frac{r v}{N})-f(x)]+\frac{1}{2^{c+1}} \nonumber\\
        & \leq \frac{2^b}{2Gr}[\frac{r}{2^b} g \cdot v +\frac{\beta (r v)^2}{2^{2b}}]
        +\frac{1}{2^{c+1}} \nonumber\\
        & \leq \frac{g \cdot v}{2G}+\frac{2^b \beta r n}{4G}+\frac{1}{2^{c+1}}.
    \end{align}
    Then,
    \begin{align}
        \|\ket{\phi}-\ket{\psi}\|^2 & = \frac{1}{2^{bn}} \sum_v | e^{2 \pi i F(v)} - e^{\frac{2 \pi i g \cdot v}{2G}} |^2 \nonumber \\
        & \leq \frac{1}{2^{bn}} \sum_v | 2 \pi i F(v) - \frac{2 \pi i g \cdot v}{2G} |^2 \nonumber \\
        & \leq \frac{1}{2^{bn}} \sum_v 4 \pi^2 (\frac{2^b \beta r n}{4G}+\frac{1}{2^{c+1}})^2.
    \end{align}
    Set $b=\log_2 \frac{G\rho}{4\pi n^2 \beta r}$, $c=\log_2{\frac{4G}{2^b n \beta r }}-1$.  We have
    \begin{align}
        \|\ket{\phi}-\ket{\psi}\|^2 \leq \frac{\rho^2}{16n^2},
    \end{align}
    which implies $\|\dyad{\phi}{\phi}-\dyad{\psi}{\psi}\|_1 \leq \frac{\rho}{2n}$. Therefore, by the union bound,
    \begin{align}
        \Pr[\abs{\frac{2^b g_i}{2G}-m_i}>n/\rho +1]<\frac{\rho}{n}, \forall{i\in [n]}.
    \end{align}
    Furthermore, there is 
    \begin{align}
        \Pr[\abs{g_i-\widetilde{\nabla}_i f(x)}>\frac{2G(n/\rho +1)}{2^b}]<\frac{\rho}{n}, \forall{i\in [n]},
    \end{align}
    as $b=\log_2 \frac{G\rho}{4\pi n^2 \beta r}$, we have
    \begin{align}
        \Pr[\abs{g_i-\widetilde{\nabla}_i f(x)}>8 \pi n^2 (n/\rho +1) \beta r/\rho]<\frac{\rho}{n}, \forall{i\in [n]}.
    \end{align}
    By the union bound, we have
    \begin{align}
        \Pr[\|{g-\widetilde{\nabla} f(x)}\|_1>8 \pi n^3 (n/\rho +1) \beta r/\rho]<\rho,
    \end{align}
    which gives the lemma.
\end{proof}

\subsection{Proof of Lemma \ref{QNSGB}}
\label{proof:QNSGB}
\QNSGB*
\begin{proof}
    Note that we omit the subscript $t$ in the proofs of the lemmas as they hold for each timestep $t$. 
    By the definition of $\alpha$-exp-concave, $e^{-\alpha f(x)}$ is a concave function. 
    For $\eta\leq \frac{\alpha}{2}$, $e^{-2\eta f(x)}$ is a concave function as well. 
    By the definition of concave functions, for any $y\in \mathcal{K}$, we have
    \begin{align}
        e^{-2\eta f(y)} & \leq e^{-2\eta f(x)} + {\nabla \left(e^{-2\eta f(x)}\right)} ^{\mathrm{T}} (y-x) \nonumber \\
        & = e^{-2\eta f(x)} + -2\eta e^{-2\eta f(x)} \nabla f(x) ^{\mathrm{T}} (y-x),
    \end{align}
    Simplifying gives
    \begin{align}
        \label{equ:expconcave}
        f(y) & \geq f(x) - \frac{1}{2\eta} \log \left(1-2\eta\nabla f(x)^{\mathrm{T}}(y-x)\right) \nonumber \\
        & \geq f(x) +  \nabla f(x)^{\mathrm{T}}(y-x)+\frac{\eta}{2} (y-x)^{\mathrm{T}} \nabla f(x) \nabla f(x)^{\mathrm{T}} (y-x), 
    \end{align}
    where the last inequality comes from $-\log(1-z)\geq z+ \frac{z^2}{4}$ for any $\abs{z}\leq 1$.
    Further, for the second term on the right hand side of Inequality (\ref{equ:expconcave}),
    \begin{align}
        \nabla f(x)^{\mathrm{T}}(y-x) & = \nabla f(x)^{\mathrm{T}}(y-x) + \widetilde{\nabla} f(x)^{\mathrm{T}}(y-x) -\widetilde{\nabla} f(x)^{\mathrm{T}}(y-x) \nonumber \\
        & = \widetilde{\nabla} f(x)^{\mathrm{T}}(y-x) + \left(\nabla f(x)-\widetilde{\nabla} f(x)\right)^{\mathrm{T}}(y-x) \nonumber \\
        & \geq \widetilde{\nabla} f(x)^{\mathrm{T}}(y-x) -\norm{\nabla f(x)-\widetilde{\nabla} f(x)}_1 \norm{y-x}_{\infty} \nonumber \\
        & \geq \widetilde{\nabla} f(x)^{\mathrm{T}}(y-x) -\norm{\nabla f(x)-\widetilde{\nabla} f(x)}_1 \norm{y-x}_2
    \end{align}
    Since $\norm{y-x}_2\leq D$, we have
    \begin{align}
        \nabla f(x)^{\mathrm{T}}(y-x) \geq \widetilde{\nabla} f(x)^{\mathrm{T}}(y-x) - D \norm{\nabla f(x)-\widetilde{\nabla} f(x)}_1.
    \end{align}
    For the third term on the right hand side of Inequality (\ref{equ:expconcave}),
    \begin{align}
        \label{equ:tvggv}
        & \frac{\eta}{2} (y-x)^{\mathrm{T}} \nabla f(x) \nabla f(x)^{\mathrm{T}} (y-x) \nonumber \\
        = & \frac{\eta}{2} (y-x)^{\mathrm{T}} \nabla f(x) \nabla f(x)^{\mathrm{T}} (y-x) + \frac{\eta}{2} (y-x)^{\mathrm{T}} \widetilde{\nabla} f(x) \widetilde{\nabla} f(x)^{\mathrm{T}} (y-x) \nonumber \\
        & - \frac{\eta}{2} (y-x)^{\mathrm{T}} \widetilde{\nabla} f(x) \widetilde{\nabla} f(x)^{\mathrm{T}} (y-x) \nonumber \\
        = & \frac{\eta}{2} (y-x)^{\mathrm{T}} \widetilde{\nabla} f(x) \widetilde{\nabla} f(x)^{\mathrm{T}} (y-x) +\frac{\eta}{2} (y-x)^{\mathrm{T}} \nabla f(x) \nabla f(x)^{\mathrm{T}} (y-x) \nonumber \\
        & - \frac{\eta}{2} (y-x)^{\mathrm{T}} \nabla f(x) \widetilde{\nabla} f(x)^{\mathrm{T}} (y-x) + \frac{\eta}{2} (y-x)^{\mathrm{T}} \nabla f(x) \widetilde{\nabla} f(x)^{\mathrm{T}} (y-x) \nonumber \\
        & - \frac{\eta}{2} (y-x)^{\mathrm{T}} \widetilde{\nabla} f(x) \widetilde{\nabla} f(x)^{\mathrm{T}} (y-x) \nonumber \\
        = & \frac{\eta}{2} (y-x)^{\mathrm{T}} \widetilde{\nabla} f(x) \widetilde{\nabla} f(x)^{\mathrm{T}} (y-x) + \frac{\eta}{2} (y-x)^{\mathrm{T}} \nabla f(x) \left(\nabla f(x) - \widetilde{\nabla} f(x)\right)^{\mathrm{T}} (y-x) \nonumber \\
        & + \frac{\eta}{2} (y-x)^{\mathrm{T}} \left(\nabla f(x)-\widetilde{\nabla} f(x) \right) \widetilde{\nabla} f(x)^{\mathrm{T}} (y-x) \nonumber \\
        \geq & \frac{\eta}{2} (y-x)^{\mathrm{T}} \widetilde{\nabla} f(x) \widetilde{\nabla} f(x)^{\mathrm{T}} (y-x) - \frac{\eta}{2} \norm{y-x}_2 \norm{\nabla f(x)}_2 \norm{\nabla f(x) - \widetilde{\nabla} f(x)}_1 \norm{y-x}_{\infty} \nonumber \\
        & - \frac{\eta}{2} \norm{y-x}_{\infty} \norm{\nabla f(x) - \widetilde{\nabla} f(x)}_1 \norm{\widetilde{\nabla} f(x)}_2 \norm{y-x}_2 \nonumber \\
        = & \frac{\eta}{2} (y-x)^{\mathrm{T}} \widetilde{\nabla} f(x) \widetilde{\nabla} f(x)^{\mathrm{T}} (y-x) - \frac{\eta}{2} \norm{y-x}_2 \norm{\nabla f(x)}_2 \norm{\nabla f(x) - \widetilde{\nabla} f(x)}_1 \norm{y-x}_{\infty} \nonumber \\
        & - \frac{\eta}{2} \norm{y-x}_{\infty} \norm{\nabla f(x) - \widetilde{\nabla} f(x)}_1 \norm{\widetilde{\nabla} f(x)-\nabla f(x)+\nabla f(x)}_2 \norm{y-x}_2 \nonumber \\
        \geq & \frac{\eta}{2} (y-x)^{\mathrm{T}} \widetilde{\nabla} f(x) \widetilde{\nabla} f(x)^{\mathrm{T}} (y-x) - \frac{\eta}{2} \norm{y-x}_2 \norm{\nabla f(x)}_2 \norm{\nabla f(x) - \widetilde{\nabla} f(x)}_1 \norm{y-x}_{\infty} \nonumber \\
        & - \frac{\eta}{2} \norm{y-x}_{\infty} \norm{\nabla f(x) - \widetilde{\nabla} f(x)}_1 \left(\norm{\widetilde{\nabla} f(x)-\nabla f(x)}_2 + \norm{\nabla f(x)}_2 \right) \norm{y-x}_2 \nonumber \\
        \geq & \frac{\eta}{2} (y-x)^{\mathrm{T}} \widetilde{\nabla} f(x) \widetilde{\nabla} f(x)^{\mathrm{T}} (y-x) - \eta \norm{y-x}_2 \norm{\nabla f(x)}_2 \norm{\nabla f(x) - \widetilde{\nabla} f(x)}_1 \norm{y-x}_{\infty} \nonumber \\
        & - \frac{\eta}{2} \norm{y-x}_{\infty} \norm{\nabla f(x) - \widetilde{\nabla} f(x)}^2_1  \norm{y-x}_2 \nonumber \\
        \geq & \frac{\eta}{2} (y-x)^{\mathrm{T}} \widetilde{\nabla} f(x) \widetilde{\nabla} f(x)^{\mathrm{T}} (y-x) - \eta \norm{y-x}_2 \norm{\nabla f(x)}_2 \norm{\nabla f(x) - \widetilde{\nabla} f(x)}_1 \norm{y-x}_2 \nonumber \\
        & - \frac{\eta}{2} \norm{y-x}_2 \norm{\nabla f(x) - \widetilde{\nabla} f(x)}^2_1  \norm{y-x}_2 \nonumber \\
        = & \frac{\eta}{2} (y-x)^{\mathrm{T}} \widetilde{\nabla} f(x) \widetilde{\nabla} f(x)^{\mathrm{T}} (y-x) - \eta  \norm{\nabla f(x)}_2 \norm{\nabla f(x) - \widetilde{\nabla} f(x)}_1 \norm{y-x}^2_2 \nonumber \\
        & - \frac{\eta}{2}  \norm{\nabla f(x) - \widetilde{\nabla} f(x)}^2_1  \norm{y-x}^2_2.
    \end{align}
    Since $\eta \leq \frac{1}{2DG},\norm{y-x}_2\leq D,\norm{\nabla f(x)}_2\leq G$, we have
    \begin{align}
        \frac{\eta}{2} (y-x)^{\mathrm{T}} \nabla f(x) \nabla f(x)^{\mathrm{T}} (y-x) \geq & \frac{\eta}{2} (y-x)^{\mathrm{T}} \widetilde{\nabla} f(x) \widetilde{\nabla} f(x)^{\mathrm{T}} (y-x) \nonumber \\
        & - \frac{D}{8} \norm{\nabla f(x) - \widetilde{\nabla} f(x)}_1 - \frac{D}{16G} \norm{\nabla f(x) - \widetilde{\nabla} f(x)}^2_1.
    \end{align}
    In summary, we have
    \begin{align}
        f_t(y)\geq & f_t(x_t)+\widetilde{\nabla} f_t(x_t)^{\mathrm{T}}(y-x_t)+\frac{\eta_t}{2}(y-x_t)^{\mathrm{T}} \widetilde{\nabla} f_t(x_t) \widetilde{\nabla} f_t(x_t)^{\mathrm{T}}(y-x_t) \nonumber \\
        & - \frac{9D}{8} \norm{\nabla f(x) - \widetilde{\nabla} f(x)}_1 - \frac{D}{16G} \norm{\nabla f(x) - \widetilde{\nabla} f(x)}^2_1.
    \end{align}
    which gives the lemma.
\end{proof}

\subsection{Proof of Theorem \ref{TNQ}}
\label{proof:TNQ}
\TNQ*
\begin{proof}
    By Lemma \ref{QGB} and Lemma \ref{QNSGB}, for any $y\in \mathcal{K}$ and any $\eta_t\leq \frac{1}{2} \min \left\{\frac{1}{GD},\alpha\right\}$, in any time step $t$, with probability greater than $1-\rho_t$, the gradient $\widetilde{\nabla} f_t(x_t)$ estimated by Algorithm \ref{QONS} satisfied
    \begin{align}
        \label{QNST1}
        f_t(y)\geq & f_t(x_t)+\widetilde{\nabla} f_t(x_t)^{\mathrm{T}}(y-x_t)+\frac{\eta_t}{2}(y-x_t)^{\mathrm{T}} \widetilde{\nabla} f_t(x_t) \widetilde{\nabla} f_t(x_t)^{\mathrm{T}}(y-x_t) \nonumber \\
        &-9D \pi n^3 (n/\rho_t +1) \beta r_t/\rho_t -\frac{4D ( \pi n^3 (n/\rho_t +1) \beta r_t/\rho_t )^2}{G} .
    \end{align}
    Inequality (\ref{QNST1}) is required to hold for all T rounds. 
    Let $B_t$ be the event that Algorithm \ref{QONS} fails to satisfy Inequality (\ref{QNST1}) in the $t$-th round.
    First, set the failure rate of each round to be the same, specifically, equal to $\rho$. Then by Lemma \ref{QNSGB}, we have $\Pr(B_1)=\Pr(B_2)=\dots=\Pr(B_T)\leq \rho$. 
    By the union bound (that is, for any finite or countable event set, the probability that at least one of the events happens is no greater than the sum of the probabilities of the events in the set), we have $\Pr(\cup_{t=1}^T B_t) \leq \sum_{t=1}^{T} \Pr(B_t) \leq T\rho$.
    Namely, the probability that Algorithm \ref{QONS} fails to satisfy Inequality (\ref{QNST1}) at least one round is less than $T\rho$, which means the probability that Algorithm \ref{QONS} succeeds for all $T$ round is greater than $1-T\rho$.
    Let $x^* \in \arg\min_{x\in \mathcal{K}} \sum_{t=1}^{T}f_t(x)$. for the fixed $y=x^*$, for all $t\in [T]$ with probability $1-T\rho$, we have
    \begin{align}
        \label{ine:T4_2}
        f_t(x_t)-f_t(x^*) & \leq \widetilde{\nabla} f_t(x_t)^{\mathrm{T}}(x_t-x^*)-\frac{\eta_t}{2}(x_t-x^*)^{\mathrm{T}} \widetilde{\nabla} f_t(x_t) \widetilde{\nabla} f_t(x_t)^{\mathrm{T}}(x_t-x^*) \nonumber \\
        &+9D \pi n^3 (n/\rho +1) \beta r_t/\rho +\frac{4D ( \pi n^3 (n/\rho +1) \beta r_t/\rho )^2}{G} .
    \end{align}
    By the update rule for $x_{t+1}$ and the Pythagorean theorem, there is
    \begin{align}
        \label{QONSTUR}
        \norm{x_{t+1}-x^*}^2_{A_t} & =\norm{\hat{P}^{(A_t)}_{\mathcal{K}}\left(x_t-\frac{1}{\eta_t} A_t^{-1} \widetilde{\nabla} f_t(x_t)\right)-x^*}^2_{A_t} \nonumber \\
        & \leq \norm{x_t-\frac{1}{\eta_t} A_t^{-1} \widetilde{\nabla} f_t(x_t)-x^*}^2_{A_t} \nonumber \\
        & \leq (x_t-x^*)^{\mathrm{T}}A_t (x_t-x^*)-\frac{2}{\eta_t} \widetilde{\nabla} f_t(x_t)^{\mathrm{T}} (x_t-x^*)+\frac{1}{\eta^2_t} \widetilde{\nabla} f_t(x_t)^{\mathrm{T}} A^{-1}_t \widetilde{\nabla} f_t(x_t).
    \end{align}
    Since $\norm{x_{t+1}-x^*}^2_{A_t}=(x_{t+1}-x^*)^{\mathrm{T}}A_t (x_{t+1}-x^*)$, combine with Inequality (\ref{QONSTUR}), we have
    \begin{align}
        \label{QONSFB}   
        \widetilde{\nabla} f_t(x_t)^{\mathrm{T}}(x_t-x^*) \leq & \frac{1}{2\eta_t} \widetilde{\nabla} f_t(x_t)^{\mathrm{T}} A^{-1}_t \widetilde{\nabla} f_t(x_t) + \frac{\eta_t}{2}(x_t-x^*)^{\mathrm{T}}A_t (x_t-x^*) \nonumber \\
        & -\frac{\eta_t}{2}(x_{t+1}-x^*)^{\mathrm{T}}A_t (x_{t+1}-x^*).
    \end{align}
    summing Inequality (\ref{QONSFB}) from $t=1$ to $T$, we have
    \begin{align}
        \sum_{t=1}^T \widetilde{\nabla} f_t(x_t)^{\mathrm{T}}(x_t-x^*) & \leq \sum_{t=1}^T \frac{1}{2\eta_t}\widetilde{\nabla} f_t(x_t)^{\mathrm{T}} A^{-1}_t \widetilde{\nabla} f_t(x_t) + \sum_{t=1}^T\frac{\eta_t}{2}(x_t-x^*)^{\mathrm{T}}A_t (x_t-x^*) \nonumber \\
        & \quad -\sum_{t=1}^T \frac{\eta_t}{2}(x_{t+1}-x^*)^{\mathrm{T}}A_t (x_{t+1}-x^*) \nonumber \\
        & =  \sum_{t=1}^T \frac{1}{2\eta_t} \widetilde{\nabla} f_t(x_t)^{\mathrm{T}} A^{-1}_t \widetilde{\nabla} f_t(x_t) + \frac{\eta_1}{2}(x_1-x^*)^{\mathrm{T}}A_1 (x_1-x^*) \nonumber \\
        & \quad + \frac{1}{2} \sum_{t=2}^T(x_t-x^*)^{\mathrm{T}}(\eta_t A_t - \eta_{t-1} A_{t-1}) (x_t-x^*) \nonumber \\
        & \quad - \frac{\eta_T}{2} (x_{T+1}-x^*)^{\mathrm{T}}A_T (x_{T+1}-x^*).
    \end{align}
    Set $\eta_1=\eta_2=\dots=\eta_T=\eta$, we have
    \begin{align}
        \label{QONSFT}
        \sum_{t=1}^T \widetilde{\nabla} f_t(x_t)^{\mathrm{T}}(x_t-x^*) & \leq \sum_{t=1}^T \frac{1}{2\eta} \widetilde{\nabla} f_t(x_t)^{\mathrm{T}} A^{-1}_t \widetilde{\nabla} f_t(x_t) + \frac{\eta}{2}(x_1-x^*)^{\mathrm{T}}A_1 (x_1-x^*) \nonumber \\
        & \quad + \frac{\eta}{2} \sum_{t=2}^T(x_t-x^*)^{\mathrm{T}}( A_t -  A_{t-1}) (x_t-x^*) \nonumber \\
        & \quad - \frac{\eta}{2} (x_{T+1}-x^*)^{\mathrm{T}}A_T (x_{T+1}-x^*) \nonumber \\
        & =  \sum_{t=1}^T \frac{1}{2\eta} \widetilde{\nabla} f_t(x_t)^{\mathrm{T}} A^{-1}_t \widetilde{\nabla} f_t(x_t) + \frac{\eta}{2}(x_1-x^*)^{\mathrm{T}}A_1 (x_1-x^*) \nonumber \\
        & \quad + \frac{\eta}{2} \sum_{t=2}^T(x_t-x^*)^{\mathrm{T}}\widetilde{\nabla} f_t(x_t)(\widetilde{\nabla} f_t(x_t))^{\mathrm{T}} (x_t-x^*) \nonumber \\
        & \quad - \frac{\eta}{2} (x_{T+1}-x^*)^{\mathrm{T}}A_T (x_{T+1}-x^*) \nonumber \\
        & =  \sum_{t=1}^T \frac{1}{2\eta} \widetilde{\nabla} f_t(x_t)^{\mathrm{T}} A^{-1}_t \widetilde{\nabla} f_t(x_t) \nonumber \\
        & \quad + \frac{\eta}{2}(x_1-x^*)^{\mathrm{T}} \left(A_1-\widetilde{\nabla} f_1(x_1)(\widetilde{\nabla} f_1(x_1))^T \right) (x_1-x^*) \nonumber \\
        & \quad + \frac{\eta}{2} \sum_{t=1}^T(x_t-x^*)^{\mathrm{T}}\widetilde{\nabla} f_t(x_t)(\widetilde{\nabla} f_t(x_t))^{\mathrm{T}} (x_t-x^*) \nonumber \\
        & \quad - \frac{\eta}{2} (x_{T+1}-x^*)^{\mathrm{T}}A_T (x_{T+1}-x^*) \nonumber \\
        & = \sum_{t=1}^T \frac{1}{2\eta} \widetilde{\nabla} f_t(x_t)^{\mathrm{T}} A^{-1}_t \widetilde{\nabla} f_t(x_t) + \frac{\epsilon \eta}{2}(x_1-x^*)^{\mathrm{T}} I (x_1-x^*) \nonumber \\
        & \quad + \frac{\eta}{2} \sum_{t=1}^T(x_t-x^*)^{\mathrm{T}}\widetilde{\nabla} f_t(x_t)(\widetilde{\nabla} f_t(x_t))^{\mathrm{T}} (x_t-x^*) \nonumber \\
        & \quad - \frac{\eta}{2} (x_{T+1}-x^*)^{\mathrm{T}}A_T (x_{T+1}-x^*) \nonumber \\
        & \leq \sum_{t=1}^T \frac{1}{2\eta} \widetilde{\nabla} f_t(x_t)^{\mathrm{T}} A^{-1}_t \widetilde{\nabla} f_t(x_t) + \frac{\epsilon \eta}{2} \norm{x_1-x^*}_2 \norm{x_1-x^*}_2 \nonumber \\
        & \quad + \frac{\eta}{2} \sum_{t=1}^T(x_t-x^*)^{\mathrm{T}}\widetilde{\nabla} f_t(x_t)(\widetilde{\nabla} f_t(x_t))^{\mathrm{T}} (x_t-x^*) \nonumber \\
        & \quad - \frac{\eta}{2} (x_{T+1}-x^*)^{\mathrm{T}}A_T (x_{T+1}-x^*) \nonumber \\
        & \leq \sum_{t=1}^T \frac{1}{2\eta} \widetilde{\nabla} f_t(x_t)^{\mathrm{T}} A^{-1}_t \widetilde{\nabla} f_t(x_t) + \frac{\epsilon \eta D^2}{2}  \nonumber \\
        & \quad + \frac{\eta}{2} \sum_{t=1}^T(x_t-x^*)^{\mathrm{T}}\widetilde{\nabla} f_t(x_t)(\widetilde{\nabla} f_t(x_t))^{\mathrm{T}} (x_t-x^*) \nonumber \\
        & \quad - \frac{\eta}{2} (x_{T+1}-x^*)^{\mathrm{T}}A_T (x_{T+1}-x^*).
    \end{align}
    For any $t\in [T]$, by the update rule, $A_t=\epsilon I + \sum_{s=1}^t \widetilde{\nabla} f_s(x_s)(\widetilde{\nabla} f_s(x_s))^{\mathrm{T}}$, since
    \begin{align}
        A_t^{\mathrm{T}}=\epsilon I^{\mathrm{T}} + \sum_{s=1}^t \left(\widetilde{\nabla} f_s(x_s)(\widetilde{\nabla} f_s(x_s))^{\mathrm{T}}\right)^{\mathrm{T}}=\epsilon I + \sum_{s=1}^t \widetilde{\nabla} f_s(x_s)(\widetilde{\nabla} f_s(x_s))^{\mathrm{T}} = A_t,
    \end{align}
    Therefore, $A_t$ are symmetric matrices, and for any $z\in \mathbb{R}^n$, 
    \begin{align}
        z^{\mathrm{T}} A_t z & = \epsilon z^{\mathrm{T}} z + \sum_{s=1}^t z^{\mathrm{T}} \widetilde{\nabla} f_s(x_s)(\widetilde{\nabla} f_s(x_s))^{\mathrm{T}} z \nonumber \\
        & = \epsilon z^{\mathrm{T}} z + \sum_{s=1}^t \left(\widetilde{\nabla} f_s(x_s))^{\mathrm{T}} z\right)^{\mathrm{T}}(\widetilde{\nabla} f_s(x_s))^{\mathrm{T}} z,
    \end{align}
    when $\epsilon \geq 0$, $z^{\mathrm{T}} A_t z$ is non-negative, namely, $A_t$ are positive semi-definite matrices. Set $\eta \geq 0, \epsilon \geq 0$, the last term in the right hand side of Inequality (\ref{QONSFT}) is negative, can be discarded. For the first term in the right hand side of Inequality (\ref{QONSFT}), we have
    \begin{align}
        \label{QONSGDB}
        \sum_{t=1}^T \frac{1}{2\eta} \widetilde{\nabla} f_t(x_t)^{\mathrm{T}} A^{-1}_t \widetilde{\nabla} f_t(x_t) & =  \frac{1}{2\eta} \sum_{t=1}^T \Tr\left(A^{-1}_t \widetilde{\nabla} f_t(x_t) \widetilde{\nabla} f_t(x_t)^{\mathrm{T}}\right) \nonumber \\
        & = \frac{1}{2\eta} \sum_{t=1}^T \Tr\left(A^{-1}_t (A_t-A_{t-1})\right) \nonumber \\
        & = \frac{1}{2\eta} \sum_{t=1}^T \Tr\left(I-A^{-0.5}_t A_{t-1} A^{-0.5}_t\right) \nonumber \\
        & = \frac{1}{2\eta} \sum_{t=1}^T \sum_{s=1}^n \left(1 - \lambda_s\left( A_t^{-0.5} A_{t-1} A_t^{-0.5} \right) \right) \nonumber \\
        & \leq - \frac{1}{2\eta} \sum_{t=1}^T \sum_{s=1}^n \log \left( \lambda_s\left( A_t^{-0.5} A_{t-1} A_t^{-0.5} \right) \right) \nonumber \\
        & = - \frac{1}{2\eta} \sum_{t=1}^T \log \left( \Pi_{s=1}^n \lambda_s\left( A_t^{-0.5} A_{t-1} A_t^{-0.5} \right) \right) \nonumber \\
        & = - \frac{1}{2\eta} \sum_{t=1}^T \log \left( \det  \left( A_t^{-0.5} A_{t-1} A_t^{-0.5} \right) \right) \nonumber \\
        & = - \frac{1}{2\eta} \sum_{t=1}^T \log \left( \det  \left( A_t^{-0.5} \right) \det \left( A_{t-1} \right) \det \left( A_t^{-0.5} \right) \right) \nonumber \\
        & = - \frac{1}{2\eta} \sum_{t=1}^T \log \left( \det  \left( A_t^{-1} \right) \det \left( A_{t-1} \right) \right) \nonumber \\
        & = - \frac{1}{2\eta} \log \left( \Pi_{t=1}^T \det  \left( A_t^{-1} \right) \det \left( A_{t-1} \right) \right) \nonumber \\
        & = - \frac{1}{2\eta} \log \left( \det  \left( A_T^{-1} \right) \det \left( A_{0} \right) \right) \nonumber \\
        & = \frac{1}{2\eta} \log \left( \det  \left( A_T \right) \det \left( A_{0}^{-1} \right) \right) \nonumber \\
        & = \frac{1}{2\eta} \log \left( \det  \left( \left( \epsilon I + \sum_{t=1}^T \widetilde{\nabla} f_t(x_t)(\widetilde{\nabla} f_t(x_t))^{\mathrm{T}} \right) (\epsilon I)^{-1} \right) \right) \nonumber \\
        & \leq \frac{1}{2\eta} \log \left(  \left( \epsilon + \sum_{t=1}^T \norm{\widetilde{\nabla} f_t(x_t)}^2_2  \right) \epsilon^{-1} \right)^n \nonumber \\
        & = \frac{n}{2\eta} \log \left(  \left( \epsilon + \sum_{t=1}^T \norm{\widetilde{\nabla} f_t(x_t)-\nabla f_t(x_t)+\nabla f_t(x_t)}^2_2  \right) \epsilon^{-1} \right) \nonumber \\
        & \leq \frac{n}{2\eta} \log \left(  \left( \epsilon + \sum_{t=1}^T \left(\norm{\widetilde{\nabla} f_t(x_t)-\nabla f_t(x_t)}_2 +\norm{\nabla f_t(x_t)}_2 \right)^2  \right) \epsilon^{-1} \right) \nonumber \\
        & \leq \frac{n}{2\eta} \log \left(  \left( \epsilon + \sum_{t=1}^T \left(\norm{\widetilde{\nabla} f_t(x_t)-\nabla f_t(x_t)}_1 + G \right)^2  \right) \epsilon^{-1} \right) \nonumber \\
        & \leq \frac{n}{2\eta} \log \left(  \left( \epsilon + \sum_{t=1}^T \left(8 \pi n^3 (n/\rho +1) \beta r_t/\rho + G \right)^2  \right) \epsilon^{-1} \right),
    \end{align}
    where $\lambda_s(A)$ is the $s$-th eigenvalue of matrix $A$, $\det(A)$ is the determinant of matrix $A$. Summing Inequality (\ref{ine:T4_2}) from $t=1$ to $T$, and combine Inequality (\ref{QONSFT}) (\ref{QONSGDB}), we have
    \begin{align}
        \sum_{t=1}^T f_t(x_t)-f_t(x^*) & \leq   \frac{n}{2\eta} \log \left(  \left( \epsilon + \sum_{t=1}^T \left(8 \pi n^3 (n/\rho +1) \beta r_t/\rho + G \right)^2  \right) \epsilon^{-1} \right)    + \frac{\epsilon \eta D^2}{2}  \nonumber \\
        & \quad +\sum_{t=1}^T 9D \pi n^3 (n/\rho +1) \beta r_t/\rho + \sum_{t=1}^T \frac{4D ( \pi n^3 (n/\rho +1) \beta r_t/\rho )^2}{G}.
    \end{align}
    Set $\eta=\min\left\{\frac{1}{2GD},\frac{\alpha}{2}\right\}, \epsilon=\frac{1}{\eta^2 D^2}, \left\{r_t=\frac{\rho G}{\pi n^3 (n/\rho +1) \beta t}\right\}_{t=1}^T$, we have
    \begin{align}
        \sum_{t=1}^T f_t(x_t)-f_t(x^*) & \leq \frac{n}{2\eta} \log \left( 1 + \sum_{t=1}^T \left( \frac{8G}{t} + G \right)^2 \eta^2 D^2    \right)    + \frac{1}{2\eta} +\sum_{t=1}^T \frac{9D G}{t} +\sum_{t=1}^T \frac{4D G^2}{t^2 G} \nonumber \\
        & \leq \frac{n}{2\eta} \log \left( 1 + \frac{1}{64} \sum_{t=1}^T \left(\frac{8}{t} + 1 \right)^2  \right)    + \frac{1}{2\eta} +\sum_{t=1}^T \frac{9D G}{t}+\sum_{t=1}^T \frac{4D G}{t} \nonumber \\
        & \leq \frac{n}{2\eta} \log \left( 1 +  \sum_{t=1}^T \left(\frac{1}{t^2} + \frac{1}{4t} + 1 \right)  \right)    + \frac{1}{2\eta} +\sum_{t=1}^T \frac{13D G}{t} \nonumber \\
        & \leq \frac{n}{2\eta} \log \left( 1 + \frac{5}{4} \log T + T  \right)    + \frac{1}{2\eta} + 13D G \log T \nonumber \\
        & \leq \left(DG+\frac{1}{\alpha}\right)n \log \left( 1 + \frac{5}{4} \log T + T  \right)    + \left(DG+\frac{1}{\alpha}\right) + 13D G \log T \nonumber \\
        & = O\left(\left(DG+\frac{1}{\alpha}\right)n\log(T+\log T)\right).
    \end{align}
    Thus, the theorem follows.
\end{proof}

\subsection{Proof of Theorem \ref{TNHQ}}
\label{proof:TNHQ}
\TNHQ*
\begin{proof}
By the assumption, let $x^* \in \arg\min_{x\in \mathcal{K}} \sum_{t=1}^{T}f_t(x)$, for the fixed $y=x^*$, for all $t\in [T]$, we have
\begin{align}
    \label{ONA}
    f_t(x_t)-f_t(x^*) \leq \nabla f_t(x_t)^{\mathrm{T}}(x_t-x^*)-\frac{\eta}{2}(x_t-x^*)^{\mathrm{T}} H(f_t)(x_t) (x_t-x^*).
\end{align}

By the update rule for $x_{t+1}$ and the Pythagorean theorem, there is
\begin{align}
        \label{ONSTUR2}
        \norm{x_{t+1}-x^*}^2_{A_t} & =\norm{\hat{P}^{(A_t)}_{\mathcal{K}}\left(x_t-\frac{1}{\eta_t} A_t^{-1} \nabla f_t(x_t)\right)-x^*}^2_{A_t} \nonumber \\
        & \leq \norm{x_t-\frac{1}{\eta_t} A_t^{-1} \nabla f_t(x_t)-x^*}^2_{A_t} \nonumber \\
        & \leq (x_t-x^*)^{\mathrm{T}}A_t (x_t-x^*)-\frac{2}{\eta_t} \nabla f_t(x_t)^{\mathrm{T}} (x_t-x^*)+\frac{1}{\eta^2_t} \nabla f_t(x_t)^{\mathrm{T}} A^{-1}_t \nabla f_t(x_t).
\end{align}

Since $\norm{x_{t+1}-x^*}^2_{A_t}=(x_{t+1}-x^*)^{\mathrm{T}}A_t (x_{t+1}-x^*)$, combine with Inequality (\ref{ONSTUR2}), we have
\begin{align}
    \label{ONSFB2}   
    \nabla f_t(x_t)^{\mathrm{T}}(x_t-x^*) \leq & \frac{1}{2\eta_t} \nabla f_t(x_t)^{\mathrm{T}} A^{-1}_t \nabla f_t(x_t) + \frac{\eta_t}{2}(x_t-x^*)^{\mathrm{T}}A_t (x_t-x^*) \nonumber \\
    & -\frac{\eta_t}{2}(x_{t+1}-x^*)^{\mathrm{T}}A_t (x_{t+1}-x^*).
\end{align}
summing Inequality (\ref{ONSFB2}) from $t=1$ to $T$, we have
\begin{align}
    \sum_{t=1}^T \nabla f_t(x_t)^{\mathrm{T}}(x_t-x^*) & \leq \sum_{t=1}^T \frac{1}{2\eta_t}\nabla f_t(x_t)^{\mathrm{T}} A^{-1}_t \nabla f_t(x_t) + \sum_{t=1}^T\frac{\eta_t}{2}(x_t-x^*)^{\mathrm{T}}A_t (x_t-x^*) \nonumber \\
    & \quad -\sum_{t=1}^T \frac{\eta_t}{2}(x_{t+1}-x^*)^{\mathrm{T}}A_t (x_{t+1}-x^*) \nonumber \\
    & = \sum_{t=1}^T \frac{1}{2\eta_t}\nabla f_t(x_t)^{\mathrm{T}} A^{-1}_t \nabla f_t(x_t) + \frac{\eta_1}{2}(x_1-x^*)^{\mathrm{T}}A_1 (x_1-x^*) \nonumber \\
    & \quad + \frac{1}{2} \sum_{t=2}^T(x_t-x^*)^{\mathrm{T}}(\eta_t A_t - \eta_{t-1} A_{t-1}) (x_t-x^*) \nonumber \\
    & \quad - \frac{\eta_T}{2} (x_{T+1}-x^*)^{\mathrm{T}}A_T (x_{T+1}-x^*).
\end{align}
Set $\eta_1=\eta_2=\dots=\eta_T=\eta$, we have
\begin{align}
    \label{ONSFT}
    \sum_{t=1}^T \nabla f_t(x_t)^{\mathrm{T}}(x_t-x^*) & \leq \sum_{t=1}^T \frac{1}{2\eta}\nabla f_t(x_t)^{\mathrm{T}} A^{-1}_t \nabla f_t(x_t) + \frac{\eta}{2}(x_1-x^*)^{\mathrm{T}}A_1 (x_1-x^*) \nonumber \\
        & \quad + \frac{\eta}{2} \sum_{t=2}^T(x_t-x^*)^{\mathrm{T}}( A_t -  A_{t-1}) (x_t-x^*) \nonumber \\
        & \quad - \frac{\eta}{2} (x_{T+1}-x^*)^{\mathrm{T}}A_T (x_{T+1}-x^*) \nonumber \\
        & = \sum_{t=1}^T \frac{1}{2\eta}\nabla f_t(x_t)^{\mathrm{T}} A^{-1}_t \nabla f_t(x_t) + \frac{\eta}{2}(x_1-x^*)^{\mathrm{T}}(A_1-H_1) (x_1-x^*) \nonumber \\
        & \quad + \frac{\eta}{2} \sum_{t=1}^T(x_t-x^*)^{\mathrm{T}} H_t (x_t-x^*) - \frac{\eta}{2} (x_{T+1}-x^*)^{\mathrm{T}}A_T (x_{T+1}-x^*) \nonumber \\
        & \leq \sum_{t=1}^T \frac{1}{2\eta}\nabla f_t(x_t)^{\mathrm{T}} A^{-1}_t \nabla f_t(x_t) + \frac{\epsilon \eta D^2}{2} + \frac{\eta}{2} \sum_{t=1}^T(x_t-x^*)^{\mathrm{T}} H_t (x_t-x^*).
\end{align}
For the first term in the right hand side of Inequality (\ref{ONSFT}), we have
\begin{align}
    \label{ONSGDB}
    \sum_{t=1}^T \frac{1}{2\eta} \nabla f_t(x_t)^{\mathrm{T}} A^{-1}_t \nabla f_t(x_t) & =  \frac{1}{2\eta} \sum_{t=1}^T \Tr\left(A^{-1}_t \nabla f_t(x_t) \nabla f_t(x_t)^{\mathrm{T}}\right) \nonumber \\
    & \leq \frac{1}{2 \alpha \eta} \sum_{t=1}^T \Tr\left(A^{-1}_t H_t\right) \nonumber \\
    & = \frac{1}{2 \alpha \eta} \sum_{t=1}^T \Tr\left(A^{-1}_t (A_t-A_{t-1})\right) \nonumber \\
    & \leq \frac{1}{2\eta} \log \left( \det  \left( A_T \right) \det \left( A_{0}^{-1} \right) \right) \nonumber \\
    & = \frac{1}{2 \alpha \eta} \log \left( \det  \left( \left( \epsilon I + \sum_{t=1}^T H_t \right) (\epsilon I)^{-1} \right) \right) \nonumber \\
    & \leq \frac{n}{2 \alpha \eta} \log \left(  \left( \epsilon + L T \right) \epsilon^{-1}\right).
\end{align}
set $\epsilon=L, \eta=\frac{1}{2DL}$. Summing Inequality (\ref{ONA}) from $t=1$ to $T$, and combine Inequality (\ref{ONSFT}) (\ref{ONSGDB}), we have
\begin{align}
    \sum_{t=1}^{T}(f_t(x_t)-f_t(x^*)) 
    & \leq \frac{n}{2 \alpha \eta} \log \left(  \left( \epsilon + L T \right) \epsilon^{-1}\right) + \frac{\epsilon \eta D^2}{2} \nonumber \\
    &\leq O\left(\frac{DLn}{\alpha}\log(T+1)\right).
\end{align}
It therefore gives the theorem.
\end{proof}

\subsection{Proof of Lemma \ref{QNSGBs}}
\label{proof:QNSGBs}
\QNSGBs*
\begin{proof}
    Note that we omit the subscript $t$ in the proofs of the lemmas as they hold for each timestep $t$. 
    By the property of $\alpha$-exp-concave
    \begin{align}
        \label{equ:expconcaveS}
        f(y) \geq f(z) +  \nabla f(z)^{\mathrm{T}}(y-z)+\frac{\eta}{2} (y-z)^{\mathrm{T}} \nabla f(z) \nabla f(z)^{\mathrm{T}} (y-z), 
    \end{align}
    For the first and the second term on the right hand side of Inequality (\ref{equ:expconcaveS}), 
    \begin{align}
         f(z)+\nabla f(z)^{\mathrm{T}}(y-z)  & =f(z)+\nabla f(z)^{\mathrm{T}}(y-z)+(\widetilde{\nabla} f(x)^{\mathrm{T}}(y-x) -\widetilde{\nabla} f(x)^{\mathrm{T}}(y-x))+(f(x)-f(x)) \nonumber \\
            & =f(x)+\widetilde{\nabla} f(x)^{\mathrm{T}}(y-x)+(\nabla f(z)-\widetilde{\nabla} f(x))^{\mathrm{T}}(y-x)  +(f(z)-f(x))+\nabla f(z)^{\mathrm{T}}(x-z) \nonumber \\
            & \geq f(x)+\widetilde{\nabla} f(x)^{\mathrm{T}}(y-x)-\|\nabla f(z)-\widetilde{\nabla} f(x)\|_1 \|y-x\|_{\infty}  -G\|z-x\|_2 - \|\nabla f(z)\|_2 \|x-z\|_2 \nonumber \\
            & \geq f(x)+\widetilde{\nabla} f(x)^{\mathrm{T}}(y-x)-\|\nabla f(z)-\widetilde{\nabla} f(x)\|_1 \|y-x\|_2  -2G\sqrt{n}r^{\prime}. \nonumber \\
            & \geq f(x)+\widetilde{\nabla} f(x)^{\mathrm{T}}(y-x)- D\|\nabla f(z)-\widetilde{\nabla} f(x)\|_1   -2G\sqrt{n}r^{\prime}.
    \end{align}
    
    Further, for the third term on the right hand side of Inequality (\ref{equ:expconcaveS}),
    \begin{align}
        & \frac{\eta}{2} (y-z)^{\mathrm{T}} \nabla f(z) \nabla f(z)^{\mathrm{T}} (y-z) \nonumber \\
        = & \frac{\eta}{2} (y-x)^{\mathrm{T}} \nabla f(z) \nabla f(z)^{\mathrm{T}} (y-x) + \frac{\eta}{2} (y-z)^{\mathrm{T}} \nabla f(z) \nabla f(z)^{\mathrm{T}} (y-z) \nonumber \\
        & - \frac{\eta}{2} (y-x)^{\mathrm{T}} \nabla f(z) \nabla f(z)^{\mathrm{T}} (y-x) \nonumber \\
        = & \frac{\eta}{2} (y-x)^{\mathrm{T}} \nabla f(z) \nabla f(z)^{\mathrm{T}} (y-x) + \frac{\eta}{2} (y-z)^{\mathrm{T}} \nabla f(z) \nabla f(z)^{\mathrm{T}} (y-z) \nonumber \\
        & - \frac{\eta}{2} (y-z)^{\mathrm{T}} \nabla f(z) \nabla f(z)^{\mathrm{T}} (y-x) + \frac{\eta}{2} (y-z)^{\mathrm{T}} \nabla f(z) \nabla f(z)^{\mathrm{T}} (y-x) \nonumber \\
        & - \frac{\eta}{2} (y-x)^{\mathrm{T}} \nabla f(z) \nabla f(z)^{\mathrm{T}} (y-x) \nonumber \\
        = & \frac{\eta}{2} (y-x)^{\mathrm{T}} \nabla f(z) \nabla f(z)^{\mathrm{T}} (y-x) + \frac{\eta}{2} (y-z)^{\mathrm{T}} \nabla f(z) \nabla f(z)^{\mathrm{T}} (x-z) \nonumber \\
        & + \frac{\eta}{2} (x-z)^{\mathrm{T}} \nabla f(z) \nabla f(z)^{\mathrm{T}} (y-x) \nonumber \\
        = & \frac{\eta}{2} (y-x)^{\mathrm{T}} \nabla f(z) \nabla f(z)^{\mathrm{T}} (y-x) + \frac{\eta}{2} (x-z)^{\mathrm{T}} \nabla f(z) \nabla f(z)^{\mathrm{T}} (x-z) \nonumber \\
        & + \frac{\eta}{2} (y-x)^{\mathrm{T}} \nabla f(z) \nabla f(z)^{\mathrm{T}} (x-z)  + \frac{\eta}{2} (x-z)^{\mathrm{T}} \nabla f(z) \nabla f(z)^{\mathrm{T}} (y-x) \nonumber \\
        \geq & \frac{\eta}{2} (y-x)^{\mathrm{T}} \nabla f(z) \nabla f(z)^{\mathrm{T}} (y-x) - \frac{\eta}{2} \norm{x-z}^2 \norm{\nabla f(z)}^2 - \eta \norm{x-z}\norm{\nabla f(z)}^2\norm{y-x} \nonumber \\
        \geq & \frac{\eta}{2} (y-x)^{\mathrm{T}} \nabla f(z) \nabla f(z)^{\mathrm{T}} (y-x) - \frac{\eta n r^{\prime 2} G^2}{2} - \eta \sqrt{n} r^{\prime} D G^2 \nonumber \\
        \geq & \frac{\eta}{2} (y-x)^{\mathrm{T}} \nabla f(z) \nabla f(z)^{\mathrm{T}} (y-x) - \frac{ n r^{\prime 2} G}{2} - \sqrt{n} r^{\prime} G.
    \end{align}
    We then handle the term $\frac{\eta}{2} (y-x)^{\mathrm{T}} \nabla f(z) \nabla f(z)^{\mathrm{T}} (y-x)$ with the same technique as shown in Inequality (\ref{equ:tvggv}), we have
    \begin{align}
        \frac{\eta}{2} (y-x)^{\mathrm{T}} \nabla f(z) \nabla f(z)^{\mathrm{T}} (y-x) \geq & \frac{\eta}{2} (y-x)^{\mathrm{T}} \widetilde{\nabla} f(z) \widetilde{\nabla} f(z)^{\mathrm{T}} (y-x) \nonumber \\
        & - \frac{D}{8} \norm{\nabla f(z) - \widetilde{\nabla} f(z)}_1 - \frac{D}{16G} \norm{\nabla f(z) - \widetilde{\nabla} f(z)}^2_1 \nonumber \\
        = & \frac{\eta}{2} (y-x)^{\mathrm{T}} \widetilde{\nabla} f(x) \widetilde{\nabla} f(x)^{\mathrm{T}} (y-x) \nonumber \\
        & - \frac{D}{8} \norm{\nabla f(z) - \widetilde{\nabla} f(z)}_1 - \frac{D}{16G} \norm{\nabla f(z) - \widetilde{\nabla} f(z)}^2_1.
    \end{align}
    
    In summary, we have
    \begin{align}
        f_t(y)\geq & f_t(x_t)+\widetilde{\nabla} f_t(x_t)^{\mathrm{T}}(y-x_t)+\frac{\eta_t}{2}(y-x_t)^{\mathrm{T}} \widetilde{\nabla} f_t(x_t) \widetilde{\nabla} f_t(x_t)^{\mathrm{T}}(y-x_t) \nonumber \\
        & - \frac{9D}{8} \norm{\nabla f(z_t) - \widetilde{\nabla} f(z_t)}_1 - \frac{D}{16G} \norm{\nabla f(z_t) - \widetilde{\nabla} f(z_t)}^2_1 -3G\sqrt{n}r^{\prime} - \frac{ n r^{\prime 2} G}{2}.
    \end{align}
    which gives the lemma.
\end{proof}

\subsection{Proof of Theorem \ref{TNQs}}
\label{proof:TNQs}
\TNQs*

\begin{proof}
    Similar to the proof of Theorem \ref{TNQ}, we have
    \begin{align}
        \sum_{t=1}^T f_t(x_t)-f_t(x^*) & \leq   \frac{n}{2\eta} \log \left(  \left( \epsilon + \sum_{t=1}^T \left(8 \pi n^3 (n/\rho +1) \beta r_t/\rho + G \right)^2  \right) \epsilon^{-1} \right)    + \frac{\epsilon \eta D^2}{2}  \nonumber \\
        & \quad +\sum_{t=1}^T 9D \pi n^3 (n/\rho +1) \beta r_t/\rho + \sum_{t=1}^T \frac{4D ( \pi n^3 (n/\rho +1) \beta r_t/\rho )^2}{G} \nonumber \\
        & \quad + \sum_{t=1}^T 3G\sqrt{n}r_t^{\prime} + \sum_{t=1}^T \frac{ n r_t^{\prime 2} G}{2}.
    \end{align}
    By Lemma \ref{Smooth}, we let $\beta=nG/p_t r^{\prime}_t$. Set $p_1=p_2=\dots=p_T=p$. 
    Then, set $\eta=\min\left\{\frac{1}{8GD},\frac{\alpha}{2}\right\}, \epsilon=\frac{1}{\eta^2 D^2}, \left\{r_t=\frac{\rho p}{\pi n^{9/2} (n/\rho +1) t^2}\right\}_{t=1}^T$, $\left\{r^{\prime}_t=\frac{1}{t\sqrt{n}}\right\}_{t=1}^T$ we have
    \begin{align}
        \sum_{t=1}^T f_t(x_t)-f_t(x^*) & \leq \frac{n}{2\eta} \log \left( 1 + \sum_{t=1}^T \left( \frac{8G}{t} + G \right)^2 \eta^2 D^2    \right)    + \frac{1}{2\eta} +\sum_{t=1}^T \frac{9D G}{t} +\sum_{t=1}^T \frac{4D G^2}{t^2 G} + \sum_{t=1}^T \frac{3G}{t} + \sum_{t=1}^T \frac{G}{2t^2} \nonumber \\
        & \leq \left(4DG+\frac{1}{\alpha}\right)n \log \left( 1 + \frac{5}{4} \log T + T  \right)    + \left(4DG+\frac{1}{\alpha}\right) + 13D G \log T + \frac{7}{2} G \log T \nonumber \\
        & = O\left(\left(DG+\frac{1}{\alpha}\right)n\log(T+\log T)\right).
    \end{align}
    Thus, the theorem follows.
\end{proof}

\end{document}